\documentclass[a4paper,UKenglish,numberwithinsect,cleveref,thm-restate]{lipics-v2021}
\pdfoutput=1 

\usepackage{mathtools} 
\usepackage{stmaryrd} 
\usepackage{centernot}
\usepackage{mdframed} 
\usepackage{multicol} 

\crefname{section}{Sect.}{Sections}
\Crefname{section}{Section}{Sections}

\theoremstyle{definition}
\newtheorem{Definition}[theorem]{Definition}

\newenvironment{Example}{\example}{\lipicsEnd\endexample}

\newcommand{\ra}{\rightarrow}
\newcommand{\E}{\exists}
\newcommand{\A}{\forall}

\renewcommand{\phi}{\varphi}
\renewcommand{\theta}{\vartheta}
\renewcommand{\emptyset}{\varnothing}
\renewcommand{\epsilon}{\varepsilon}
\renewcommand{\AA}{{\mathfrak A}}

\newcommand{\FO}{{\rm FO}}

\DeclareMathOperator{\Str}{\mathrm{Str}}

\DeclareMathOperator{\Lit}{\mathrm{Lit}}

\newcommand*{\from}{\colon}
\newcommand*{\without}[1]{\setminus\{#1\}}

\newcommand*{\tup}[1]{\mathbf{#1}}
\newcommand{\ta}{\tup a}
\newcommand{\tb}{\tup b}
\newcommand{\tx}{\tup x}

\newcommand{\tz}{\tup z}

\newcommand{\Bool}{\mathbb{B}}

\newcommand{\N}{{\mathbb N}}
\newcommand{\Ninf}{{\mathbb N}^{\infty}}
\newcommand{\Sinf}{{\mathbb S}^{\infty}}
\newcommand{\Trop}{\mathbb{T}}
\newcommand{\Vit}{\mathbb{V}}
\newcommand{\Luk}{\mathbb{L}}

\newcommand*{\Real}{\mathbb{R}}

\newcommand*{\ext}[1]{[\![ #1 ]\!]}
\newcommand*{\simbool}{=_{\mathbb B}}

\newcommand{\Kmin}[1][K]{{#1}_{\text{inf}}}
\newcommand{\pimin}{\pi_{\text{inf}}}

\newcommand{\Pub}{\mathsf{P}}
\newcommand{\Cnf}{\mathsf{C}}
\newcommand{\Sec}{\mathsf{S}}
\newcommand{\Tsec}{\mathsf{T}}

\newcommand{\Inf}{\bigsqcap}
\newcommand{\Sup}{\bigsqcup}
\newcommand{\meet}{\sqcap}
\newcommand{\join}{\sqcup}

\renewcommand{\bar}{\overline}
\newcommand*{\nn}[1]{\bar{#1}}
\newcommand{\nnX}{\nn{X}}
\newcommand{\bX}{\mathbf {X}}
\newcommand{\bY}{\mathbf {Y}}

\newcommand{\Ane}{\A^{\neq}}
\newcommand{\Ene}{\E^{\neq}}
\newcommand{\XX}[1]{\bX^{(#1)}}
\newcommand{\YY}[1]{\bY^{(#1)}}
\newcommand{\eps}{\epsilon}
\DeclareMathOperator{\arity}{arity}
\DeclareMathOperator{\as}{ASV}
\DeclareMathOperator{\extend}{ext}

\newcommand{\IF}{\textbf{if}\ }
\newcommand{\THEN}{\textbf{then}\ }
\newcommand{\ELSE}{\textbf{else}\ }
\newcommand{\DO}{\textbf{do}\ }

\newcommand{\GUESS}{\textbf{guess}\ }
\newcommand{\CHOOSE}{\textbf{universally choose}\ }
\newcommand{\REJECT}{\textbf{reject}\ }
\newcommand{\ACCEPT}{\textbf{accept}\ }

\newcommand{\cxclass}[1]{\textsc{#1}}

\newcommand{\pspace}{\cxclass{Pspace}}

\newcommand{\rhopi}{\rho^\pi}
\newcommand{\pol}[2][]{f_{#2}^{#1}} 
\newcommand{\poly}[3][]{f_{#2}^{#1}[#3]} 
\newcommand{\polyrho}[3][]{\poly[#1]{#2}{\rhopi_{#3}}} 
\newcommand{\fml}[1]{\pi \ext {#1}} 
\newcommand{\fmla}[2]{\pi \ext {#1(#2)}} 
\newcommand{\expa}[2]{g_{#1}[#2]} 
\newcommand{\exparho}[2]{\expa{#1}{\rhopi_{#2}}} 

\title{Zero-One Laws and Almost Sure Valuations of First-Order Logic in Semiring Semantics}
\titlerunning{Zero-One Laws in Semiring Semantics}

\author{Erich Grädel}{RWTH Aachen University, Germany}{graedel@logic.rwth-aachen.de}{}{}
\author{Hayyan Helal}{RWTH Aachen University, Germany}{hayyan.helal@rwth-aachen.de}{}{}
\author{Matthias Naaf}{RWTH Aachen University, Germany}{naaf@logic.rwth-aachen.de}{}{}
\author{Richard Wilke}{RWTH Aachen University, Germany}{wilke@logic.rwth-aachen.de}{}{}
\authorrunning{E.~Grädel, H.~Helal, M.~Naaf, and R.~Wilke}


\nolinenumbers
\hideLIPIcs

\begin{document}

\maketitle

\begin{abstract}
Semiring semantics evaluates logical statements by values in some commutative semiring
$(K,+,\cdot,0,1)$.
Random semiring interpretations, induced by a probability distribution on $K$, generalise random structures, and we investigate
here the question of how classical results on first-order logic on random structures, 
most importantly the 0-1 laws of Glebskii et al. and Fagin, generalise to semiring semantics.
For positive semirings, the classical 0-1 law implies that every first-order sentence is, 
asymptotically, either almost surely evaluated to 0 by random semiring interpretations, or almost surely 
takes only values different from 0. However, by means of a more sophisticated analysis,
based on appropriate extension properties and on algebraic representations of first-order formulae, 
we can prove much stronger results.

For many semirings $K$
the first-order sentences in $\FO(\tau)$ can be partitioned into classes 
 $(\Phi_j)_{j\in K}$ such that for each $j\in K$, every sentence in $\Phi_j$ evaluates 
almost surely to $j$ under random semiring interpretations. Further, 
for finite or infinite lattice semirings,
this partition actually collapses to just three classes
$\Phi_0$, $\Phi_1$, and $\Phi_\epsilon$, of sentences that, respectively, almost surely 
evaluate to 0, 1, and to the smallest value $\epsilon\neq 0$.
For all other values $j\in K$ we have that
$\Phi_j=\emptyset$. The problem of computing the almost sure valuation of a first-order sentence on
finite lattice semirings is \pspace-complete.

An important semiring where the analysis is somewhat different is the \emph{natural semiring} $(\N,+,\cdot,0,1)$.
Here, both addition and multiplication are increasing with respect to the natural semiring order and the classes
$(\Phi_j)_{j\in \N}$ no longer cover all $\FO(\tau)$-sentences, but have to
be extended by $\Phi_\infty$, the class of sentences that almost surely evaluate to
unboundedly large values. 
\end{abstract}

\newpage
\section{Introduction}

Semiring semantics is based on the idea to evaluate logical statements not just by \emph{true}
or \emph{false}, but by values in some commutative semiring $(K,+,\cdot,0,1)$.
In this context, the standard semantics appears as the special case 
when the Boolean semiring $\Bool = (\{\bot, \top\}, \lor, \land, \bot, \top)$ is used.
Valuations in other semirings provide additional information, beyond the truth
or falsity of a statement:
the Viterbi-semiring $\Vit = ([0,1]_{\Real}, \max, \cdot, 0, 1)$ models \emph{confidence scores}, 
the tropical semiring $\Trop= (\Real_{+}^{\infty}, \min, +, \infty, 0)$
is used for \emph{cost analysis}, and min-max-semirings $(K, \max, \min, a, b)$ for a 
totally ordered set $(K,<)$ can model, for instance, different \emph{access levels}.
Further, semirings of polynomials or formal power series  permit us to \emph{track} which
atomic facts are used (and how often) to establish the truth of a sentence in a given structure.

Semiring semantics originated in the provenance analysis for (positive) database query languages, such as positive relational algebra or datalog
(see e.g. \cite{GreenTan17,Glavic21} for surveys), but in the last years it has been systematically extended to many logical systems,
including first-order logic, modal logic, description logics, and fixed-point logic. This raises the question to what extent the standard results and techniques of classical logic (and specifically finite model theory) extend to semiring semantics, and
how such extensions depend on the choice of the underlying semiring. This paper is part of a systematic study of
model theoretic methods for semiring semantics, and it is devoted to the study of 0-1 laws in semiring semantics.

\medskip
We briefly recall some basic facts about 0-1 laws for first-order logic on  random structures.
For a finite relational vocabulary $\tau$, a finite universe $[n]=\{0,\dots,n-1\}$, a constant $p$ with
$0<p<1$, we consider the probability spaces $\Str_{n,p}(\tau)$ of random $\tau$-structures with
universe $[n]$, obtained by the random experiment which, independently for
each relational atom $\alpha=Ri_1\dots i_k$ (where $R\in\tau$ has arity $k$, and $i_1,\dots i_k\in[n]$),
makes a random choice whether $\alpha$ shall be true (with probability $p$), or false (with probability $1-p$).
The most common such distribution is the \emph{uniform} one, with $p=1/2$, which gives to each possible
$\tau$-structure over $[n]$ the same probability.  Beyond the case where $p$ is constant, there has also
been intensive research on probability spaces $\Str_{n,p}(\tau)$, where the probabilities of
atomic facts depend on the size of the universe, i.e.\ are given by a function $p\from\N\to[0,1]$; however, in this paper we
will consider only atomic probabilities that are the same for each $n$.

Given a first-order sentence $\psi\in\FO(\tau)$ we define $\mu_{n,p}(\psi)$ to be the probability
that a random structure from $\Str_{n,p}(\tau)$ is a model of $\psi$, and we are interested in
the behaviour of the sequence $(\mu_{n,p}(\psi))_{n\in\N}$ as $n$ tends to infinity. A fundamental result,
proved in \cite{GlebskiiKogLioTal69} and \cite{Fagin76} is the celebrated 0-1 law for first-order logic:

\begin{theorem} For every sentence $\psi\in\FO(\tau)$ the asymptotic probability
$\lim_{n\to\infty} \mu_{n,p}(\psi)$ exists, and is either 0 or 1. Moreover, the sequence 
$(\mu_{n,p}(\psi))_{n\in\N}$ converges exponentially fast to this limit.
\end{theorem}

Informally, we say that each sentence $\psi\in\FO(\tau)$ is almost surely true or almost surely 
false on finite structures. There are several possibilities to prove the 0-1 law. In the original proof
of Glebskii et al.~\cite{GlebskiiKogLioTal69} a quantifier elimination argument was used. 
Later, Fagin \cite{Fagin76} presented a different proof based on the theory of 
\emph{extension axioms} which, informally,  say that every configuration of $k$ points can be extended
in any consistent way to a configuration of $k+1$ points. For undirected graphs, for instance, this means that
for any collection $v_1,\dots,v_k,$ of $k$ nodes and any $i\leq k$ there is a further node $w$ which
is adjacent to $v_1,\dots,v_i$, but not to $v_{i+1},\dots,v_k$. Fagin's proof relies on the
following facts:
\begin{itemize}
\item Each extension axiom is almost surely true on random structures.
\item The theory $T$ of all extension axioms is $\omega$-categorical, i.e.\ it has a unique
countable model, up to isomorphism, which is sometimes called
the \emph{random $\tau$-structure} or, in the case of undirected graphs, the \emph{Rado graph}.
\item It follows that $T$ is complete, i.e.~either $T\models\psi$ or $T\models \neg\psi$,
for every sentence $\psi\in\FO(\tau)$.
By compactness it then follows that either $\psi$ or $\neg\psi$ is a consequence of finitely
many extension axioms, and is therefore almost surely true on random $\tau$-structures.
\item Moreover, it follows that $\psi$ is almost surely true on finite $\tau$-structures if, and only if, $\psi$ is true
in the countable random $\tau$-structure.
\end{itemize}

\medskip
The 0-1 law has been extended in many directions, to more powerful logics than $\FO$ \cite{KolaitisVar90,KolaitisVar92}, 
to different probability distributions \cite{Spencer93},  to more general  kinds of convergence laws,
and to specific classes of mathematical structures (see \cite{Compton89} for a survey). Such results often give a simple and direct argument for proving that properties for which these convergence laws fail cannot be expressed
in such logics. A simple and well-known example is the fact that no first-order sentence (and in fact, no sentence in
bounded-variable infinitary logic) can distinguish between finite structures of even and odd cardinality. 
More practically, 0-1 laws have also been put to use for studying query answering in the context of
uncertain data (see e.g. \cite{Libkin18}).
It is a natural question whether semiring semantics admits results that are analogous to the 0-1 law
of finite model theory. Fixing a probability distribution over a given semiring $K$, the notion of a random structure
generalizes in a rather straightforward way to the notion of a random $K$-interpretation, so the typical questions
studied for logic on random structures make perfectly sense in the context of semiring semantics.
Notice that the 0-1 law splits the relational first-order sentences
into two classes: those that are almost 
surely true and those that are almost surely false (on finite structures). Is there a similar split for valuations in other
semirings than the Boolean one? For instance, given a finite semiring $K$, can we partition $\FO(\tau)$
into classes $(\Phi_j)_{j\in K}$ such that for each $j\in K$, every sentence in $\Phi_j$ evaluates 
almost surely to $j$, under random semiring interpretations into $K$?
And are all the classes $\Phi_j$ non-empty, or do the almost sure valuations concentrate on just a few values,
for instance on 0 and 1? 
How are these partitions, if they exist, related if we compare different semirings?
For instance, are the almost surely false sentences always the same, no matter which semiring we
consider? 
Are there similar results for infinite semirings? More generally, what kind of algebraic conditions do we have 
to impose on the underlying semiring to obtain results that are analogous to the traditional 0-1 law?
Finally, there also are questions of complexity: how difficult is it to compute the
almost sure valuation of a given first-order sentence (assuming that it exists)? 
Besides the intrinsic mathematical interest as a fundamental model-theoretic issue
about semiring semantics, such results have the potential to lead to more general applications
than the classical 0-1 laws, 
concerning for instance the (non-)definability of numerical parameters of graphs and other structures, or
the provenance analysis for uncertain or probabilistic data.

\medskip
Our methods to answer such questions combine on the one hand techniques that are adapted from 
traditional studies of logic on random structures, such as extension properties of atomic types, 
and on the other side specific ideas of semiring semantics, such as the use
of polynomials with indeterminates for tracking the literals. Our methods work best for
\emph{absorptive semirings}; these are semirings that are naturally ordered, in the sense that 
$a\leq b:\Leftrightarrow \E c (a+c=b)$
is a partial order, and that multiplication is decreasing with respect to this order.
This is not a very serious restriction as most of the application semirings used in provenance analysis
(such as the Viterbi semiring, the tropical semirings, the \L{}ukasiewicz semiring, all lattice semirings etc.) are indeed absorptive,
and absorptive semirings have turned out to be relevant also for the 
analysis of fixed-point logics \cite{DannertGraNaaTan21} and infinite games \cite{GraedelLucNaa21,GraedelTan20}.

\medskip
The general picture that emerges from our analysis of random $K$-interpretations for a number of
different semirings $K$ shows that there indeed is a 0-1 law, saying that with probabilities converging to 1 exponentially fast,  the valuation $\pi\ext{\psi}$ of a first-order sentence $\psi$ almost surely concentrates on one  
specific value $j\in K$. While this is perhaps not really surprising, given the 0-1 law for the classical Boolean semantics,
the analysis of the induced partition $(\Phi_j)_{j\in K}$ of $\FO(\tau)$
into classes of sentences that almost surely evaluate to $j$, is rather interesting.
It neither is the case that all semiring elements $j\in K$ appear as almost sure values of first-order sentences,  nor 
that these concentrate exclusively on the smallest and largest values 
(i.e.\ 0 and 1 in absorptive semirings). For finite and infinite lattice semirings, we show, by means of a 
description of first-order
formulae  by polynomials, that there is a partition of $\FO(\tau)$ into
three classes $\Phi_0$, $\Phi_1$, and $\Phi_\epsilon$, of sentences that, respectively, almost surely 
evaluate to 0, 1, and to $\epsilon = \Inf \{ j \in K : j \neq 0 \}$.
Notice that $\epsilon$ is the smallest element greater than $0$, if such an element exists
(as for instance in finite min-max semirings). For all other values $j\in K$ we have that
$\Phi_j=\emptyset$. Over most semirings the three classes $\Phi_0$, $\Phi_1$, $\Phi_\epsilon$ are distinct, but there are a few cases
where we have only two classes because  $\Phi_\epsilon$ collapses to $\Phi_0$ (as in distributive lattices without a smallest positive element), or to $\Phi_1$ (in the Boolean semiring).

An important semiring where the analysis is somewhat different is the \emph{natural semiring} $(\N,+,\cdot,0,1)$; this semiring is not absorptive and multiplication is increasing.
The 0-1 law still holds for the natural semiring, but the proof relies on more general $\infty$-expressions instead of polynomials and there are rather trivial constructions showing that every number $j\in\N$ appears as almost sure valuation.
We show that in fact all sentences with almost sure valuations in $\N \without 0$ are `trivial',
or in other words, the `interesting' first-order sentences are almost surely false or almost surely have arbitrarily
large truth values on random $\N$-interpretations.

\section{Semiring semantics}

For a finite relational vocabulary  $\tau$, we write $\Lit_k(\tau)$ for the set of atoms $R\tup z$ 
and negated atoms $\neg R\tup z$ with $R \in \tau$ and where $\tup z$ is any tuple of variables taken
from $\{x_1, \dots, x_k\}$. For a universe $A$, we write $\Lit_A(\tau)$ for the set of
\emph{instantiated} $\tau$-literals $R\ta$ and $\neg R\ta$ with $\ta \in A^{\arity(R)}$.
We interpret these literals by values in a commutative%
\footnote{In the following, \emph{semiring} always refers to a commutative semiring.}
semiring, which is an algebraic structure $(K,+,\cdot,0,1)$ with $0\neq1$, such that $(K,+,0)$
and $(K,\cdot,1)$ are commutative monoids, $\cdot$
distributes over $+$, and $0\cdot a=a\cdot 0=0$.

Given a commutative semiring $K$, a \emph{$K$-interpretation} (of vocabulary $\tau$ and universe $A$)
is a function $\pi\from\Lit_A(\tau)\to K$.
We only consider $K$-interpretations which are \emph{model-defining}: 
for any pair of complementary literals $\alpha, \neg\alpha$
precisely one of the values $\pi(\alpha)$, $\pi(\neg\alpha)$ is 0.

A $K$-interpretation $\pi \from \Lit_A(\tau) \to K$ extends in a straightforward way
to a valuation of any instantiation $\phi(\ta)$ of a formula $\phi(\tx) \in \FO(\tau)$, 
assumed to be written in negation normal form,
by a tuple $\ta \subseteq A$.
The semiring semantics $\pi \ext{\phi(\ta)}$ is defined
by induction. We first extend $\pi$ by mapping equalities and inequalities to their truth values by 
\[\pi \ext{a = b} \coloneqq\begin{cases} 1 &\text{ if } a = b \\0 &\text{ if } a \neq b \end{cases}\quad\text{and}\quad
 \pi \ext{a \neq b} \coloneqq\begin{cases} 0 &\text{ if } a = b \\1 &\text{ if } a \neq b \end{cases},\]
and by interpreting disjunctions and existential quantifiers as sums, and conjunctions and universal quantifiers as products:
\begin{alignat*}{3}
\pi \ext{\psi(\ta) \lor \theta(\ta)} &\coloneqq \pi \ext{\psi(\ta)} + \pi \ext{\theta(\ta)} &\quad\quad\quad \pi \ext{\psi(\ta) \land \theta(\ta)} &\coloneqq \pi \ext{\psi(\ta)} \cdot \pi \ext{\theta(\ta)} \\
\pi \ext{\exists x \, \theta(\ta, x)} &\coloneqq \sum_{a \in A} \pi \ext{\theta(\ta, a)} &\quad\quad\quad \pi \ext{\forall x \, \theta(\ta, x)} &\coloneqq \prod_{a \in A} \pi \ext{\theta(\ta, a)}.
\end{alignat*}

For the treatment of extension properties of atomic types and 0-1 laws some minor headaches
in the form of necessary case distinctions can be caused by equalities and inequalities.
To simplify our proofs we thus rewrite first-order formulae by means of  the \emph{excluding quantifiers}
$\Ene$ and $\Ane$ with the equivalences (in Boolean as well as semiring semantics) that
for any formula $\phi(\tx,y)$, with $\tx=(x_1,\dots,x_k)$ and free variables as displayed,
\[
  \Ene y\, \phi(\tx,y) \equiv \E y(\bigwedge_{i=1}^k y\neq x_i \land \phi(\tx,y)) \quad\text{ and }\quad
  \Ane y\, \phi(\tx,y) \equiv \A y(\bigvee_{i=1}^k y= x_i \lor  \phi(\tx,y)).
\]
Clearly, the classical quantifiers  $\E$ and $\A$ can be expressed (again in Boolean as well as semiring semantics)
by
$\E y\, \phi(\tx,y) \equiv \bigvee_{i=1}^k  \phi(\tx,x_i) \lor\Ene y\,  \phi(\tx,y)$ and
$\A y\, \phi(\tx,y) \equiv \bigwedge_{i=1}^k  \phi(\tx,x_i) \land \Ane y\,  \phi(\tx,y)$.

\section{Random semiring interpretations}\label{sect:random}

For a universe $A$, a relational vocabulary $\tau$ and a commutative semiring $K$, we denote by
$K\text{-Int}_A[\tau]$ the set of
$K$-interpretations $\pi\from\Lit_A(\tau)\to K$ on universe $A$.
Given a probability measure $\mu$ on $K\text{-Int}_A[\tau]$, a sentence $\psi\in\FO(\tau)$
and a set of semiring values $J\subseteq K$ let
\[   \mu [ \pi\ext{\psi} \in J ]\coloneqq \mu \{\pi\in K\text{-Int}_A[\tau] :  \pi\ext{\psi}\in J \}.\]

The probability measures $\mu_{n,p}$ we are interested in are obtained by choosing semiring values
for the literals over the universe $A=[n]$ independently and at random, keeping in mind that for
complementary literals $\alpha$ and $\neg\alpha$ precisely one should
get the value 0, and the other one an arbitrary non-zero value.  
Given a probability distribution $p$ on $K^+ \coloneqq K\without 0$,
a random $K$-interpretation $\pi$ thus makes, independently for each relational atom $\alpha=R\ta$, 
a random choice with probability $1/2$ whether $\alpha$ or $\neg\alpha$ shall be true; if $\alpha$ is true, then set
$\pi(\neg\alpha)=0$ and select for $\pi(\alpha)$ a random value from $K^+$
according to $p$;
analogously, if  $\alpha$ is false, then we set
$\pi(\alpha)=0$ and select  $\pi(\neg\alpha)\in K^+$ at random.
Every  $K$-interpretation $\pi$ chosen in this way is \emph{model-defining}.
For finite semirings $K$, the most natural probability distribution on $K^+$ is the uniform one,
so that the probability that $\pi(\alpha)=j$ is  $1/2(|K|-1)$ for any $j\neq 0$.
But our results hold for all measures $\mu_{n,p}$
as long as the choices whether $\alpha$ or $\neg\alpha$
are done with a constant probability (not necessarily 1/2) and all semiring values occur with positive probability, i.e., $p \from K^+ \to (0,1]$.
For fixed $p$, $\psi$ and $j\in K$ we then
consider the sequence $(\mu_{n,p}[ \pi\ext{\psi}=j])_{n<\omega}$ of probabilities
that $\psi$ evaluates to the semiring value $j$ in a random $K$-interpretation on universe $[n]$ (with positive semiring values chosen according to the probability distribution $p$).

\begin{Definition} 
We say that a \emph{0-1 law} holds for a class of sentences $\Phi$,
a finite or countable semiring $K$ and a probability distribution $p$ on $K^+$,
if for each sentence $\psi\in\Phi$ and each value $j\in K$
the sequence $(\mu_{n,p}[\pi\ext{ \psi}=j])_{n<\omega}$ converges to either 0 or 1, as $n$ goes to infinity.
In that case, we denote by  $\as_{K,p}(\psi)$ the \emph{almost sure valuation} of $\psi$ for $K$ and $p$,
defined as the unique value $j\in K$ such that 
$\lim_{n\to\infty} \mu_{n,p}[ \pi\ext{\psi}=j]= 1$.
Further, let $\as_{K,p}(\Phi)$ the set of possible almost sure valuations that may appear for
sentences in $\Phi$, formally, $\as_{K,p}(\Phi)\coloneqq \{ \as_{K,p}(\psi): \psi\in\Phi\}$.
\end{Definition} 

\medskip
Later we shall also study  semirings over real numbers, such as the Viterbi semiring $\Vit$, the tropical semiring $\Trop$,
or the min-max semiring over the real interval $[0,1]$. For these, the definitions have to be 
adjusted somewhat. They are discussed in \cref{sec:infinite} below.

\section{Extension properties}

Similar to the Boolean case, we study configurations of $k$ points (which we always assume to be pairwise distinct) and whether they can be extended to $k+1$ points.

\begin{Definition} An \emph{atomic $k$-type} (of vocabulary $\tau$ in the semiring $K$) is a consistent 
valuation $\rho\from\Lit_k(\tau)\to K$, consistent in the sense that for every $\tau$-atom $\alpha$, precisely one of $\rho(\alpha)$, $\rho(\neg\alpha)$ is $0$.
Given a $K$-interpretation $\pi\from\Lit_A(\tau)\to K$, a
tuple $\ta=(a_1,\dots,a_k)$ of pairwise distinct elements induces the $k$-type $\rho^\pi_{\ta}$ that maps all literals $\beta\in\Lit_k(\tau)$ to $\pi(\beta[\ta])$, where $\beta[\ta]$ is the instantiation of variables $x_i$ by $a_i$.
\end{Definition}

For two $k$-types $\rho,\rho' \from \Lit_k(\tau)\to K$, we write $\rho\simbool \rho'$ if
$\rho$ and $\rho'$ map precisely the same literals to 0 (thus inducing
the same Boolean type when identifying all non-zero values).
Further let $\rho\leq\rho'$ if $\rho(\beta)\leq\rho(\beta')$ for all $\beta\in\Lit_k(\tau)$. Notice that this can only be 
the case if $\rho\simbool\rho'$; indeed if 
$0=\rho(\beta)< \rho'(\beta)\neq 0$ then $0 \neq \rho(\neg\beta)\not\leq \rho'(\neg\beta)=0$.

We say that a tuple $\ta$ of pairwise distinct elements \emph{realises the atomic $k$-type} $\rho$ in $\pi$,
if $\rho^\pi_{\ta}=\rho.$ 
For $k>m$ an atomic $k$-type $\rho^+$ \emph{extends} the atomic $m$-type $\rho$
if $\rho^+\restriction \Lit_m(\tau) =\rho$. 
In that case, every realisation $(a_1,\dots, a_k)$ of $\rho^+$ in $\pi$ restricts to a 
realisation $(a_1,\dots,a_m)$ of $\rho$. 
On the other side, it is not clear whether a tuple that realises $\rho$ can be extended 
to a realisation of $\rho^+$.
We formulate extension properties that guarantee the existence of such extensions, which 
play a central role in the proof of 0-1 laws.  
Given an atomic $m$-type $\rho$, let  $\extend(\rho)$ be the set of atomic $(m+1)$-types that extend $\rho$.

\begin{Definition} A $K$-interpretation $\pi$ has the \emph{$k$-extension property}
if for every $m< k$, every atomic $m$-type $\rho$ and every extension $\rho^+\in\extend(\rho)$,
the following holds: every tuple $\ta$ that realises  $\rho$ in $\pi$ 
can be extended to a realisation $(\ta,b)$ of  $\rho^+$, for some $b \in A \setminus \ta$.
\end{Definition} 
 
\begin{proposition}\label{extension-property}  Fix a finite relational vocabulary $\tau$ and
let $K$ be a countable semiring with a probability distribution $p\from K^+ \to(0,1]$. 
For every atomic $m$-type $\rho$ and every extension $\rho^+\in\extend(\rho)$,
\[\lim_{n\to\infty} \mu_{n,p}[ \text{  every realisation of $\rho$ in $\pi$ 
can be extended to a realisation of $\rho^+$}]  = 1,\]
and the convergence to this limit is exponentially fast.
For any finite semiring, we moreover have, again with exponential convergence, that
random $K$-interpretations almost surely have the $k$-extension property (for any fixed $k$).
\end{proposition}

\begin{proof}
We first calculate, for any given $(m+1)$-type $\rho^+$
and its restriction $\rho=\rho^+\restriction \Lit_m(\tau)$, a  bound 
for the probability that
a random $K$-interpretation on $n$ elements has some realisation $\ta$ of $\rho$ 
that can \emph{not} be extended to a realisation
$(\ta,b)$ of $\rho^+$. There is a fixed collection $\alpha_1,\dots,\alpha_q$ of relational atoms in 
$\Lit_{m+1}(\tau)$ in which the variable $x_{m+1}$ occurs; hence there is a
fixed collection $(s_1,r_1),\dots,(s_q,r_q)$ of elements of $\{\bot,\top\}\times K^+$, where
\[  (s_i,r_i)=\begin{cases} (\top,j), &\text{ if }\rho^+(\alpha_i)=j \text{ and } \rho^+(\neg\alpha)=0,\\
(\bot,j), &\text{ if }\rho^+(\alpha_i)=0 \text{ and } \rho^+(\neg\alpha)=j.\end{cases}\]
Thus, the probability that all values chosen by a random $K$-interpretations coincide with
those required by $\rho^+$ is
\[  f(\rho^+)\coloneqq 2^{-q}\prod_{i=1}^q p(r_i) > 0.\]
Thus, for any given realisation $\ta=(a_1,\dots,a_m)$ of $\rho$,
the probability that a fixed $b\in[n]\setminus\{a_1,\dots,a_m\}$ does \emph{not} provide
a realisation $(\ta,b)$ of $\rho^+$ is $(1-f(\rho^+))$. It follows that
\[  \mu_{n,p}[ \text{ some realisation of $\rho$ does not extend to a realisation of $\rho^+$} ] \leq n^m(1-f(\rho^+))^{n-m}\]
which for growing $n$ converges to 0 exponentially fast.

Over an infinite semiring there exist infinitely many atomic $k$-types for any $k\geq 1$, so  we cannot realise all
of them on a finite universe. Thus  $\mu_{n,p}[ \pi\text{ has the $k$-extension property }] = 0$ for all $n$.
However, over a finite semiring, each $k$ admits only a bounded number of atomic $k$-types, and we conclude that $\lim_{n\to\infty} \mu_{n,p}[ \pi\text{ has the $k$-extension property }] = 1$.
\end{proof}

\section{First-order formulae and semirings of polynomials}

By \cite{GraedelTan17}  we can describe the semiring semantics of  any first-order 
sentence $\psi\in\FO(\tau)$ on a finite universe $A$
by a polynomial $f^A_\psi \in \N[\XX A]$  in the set of indeterminates $\XX A$, which has, for every fully instantiated $\tau$-atom $R\ta \in \Lit_A(\tau)$ over $A$, two indeterminates
$X_{R\ta}$ and $X_{\neg R\ta}$.
For any $K$-interpretation $\pi\from\Lit_A(\tau)\to K$, we have that $\pi\ext\psi=f^A_\psi[\pi]$,
where $f^A_\psi[\pi]$
results from $f^A_\psi$ by substituting the indeterminate $X_\beta$ by $\pi(\beta)$, for every literal $\beta \in \Lit_A(\tau)$.
Clearly, the set $\XX A$, and hence the polynomial $f^A_\psi$, depends on $A$.

We shall prove that for semiring interpretations with the $k$-extension property, we can do
better. For any natural number $i$, let  $\XX i$ be the set
of indeterminates $X_\alpha$ and $X_{\neg\alpha}$
for $\tau$-atoms $\alpha = R\tz \in \Lit_i(\tau)$ using only variables $x_1,\dots,x_i$.
Notice that $\XX i$ depends only on $i$ and $\tau$, but not on the universe.
Further, let $E=(\{0,e,1\},+,\cdot,0,1)$ 
be the three-element semiring  with $e+e=e\cdot e=e$ and $e+1=1$.
We describe any formula $\psi(x_1,\dots,x_i)$ with $i\leq k$ by a formal polynomial
$f_\psi\in E[\XX i]$, independent of the size of the universe on which we evaluate $\psi$.
As usual, we can write $ f_\psi = m_1 + \dots + m_l$
as a sum of monomials of the form $m = cX_1\dots X_\ell$ in indeterminates from $\XX i$ and with coefficient $c\in\{0,1,e\}$.

\begin{Definition} \label{def-polynomial} Let $\psi(x_1,\dots,x_i)\in\FO(\tau)$ for a finite relational
vocabulary $\tau$. Recall that we assume that $\psi$ is in negation normal form
and written with the excluding quantifiers $\Ene$ and $\Ane$.
The associated polynomial $f_\psi(\XX i)$ is defined by
induction, as follows.
\begin{itemize}
\item If $\psi$ is an equality $x_j=x_\ell$ then $f_\psi=1$ if $j=\ell$ and $f_\psi=0$ if $j\neq \ell$.
Similarly, if $\psi$ is an inequality $x_j\neq x_\ell$ then $f_\psi=1$ if $j\neq \ell$ and $f_\psi=0$ if $j=\ell$.

\item If $\psi$ is a relational atom $\alpha$ or its negation $\neg \alpha$, then $f_\psi=X_\alpha$ or $f_\psi=X_{\neg \alpha}$, respectively.

\item For disjunctions and conjunctions, we set $f_{\psi\lor\phi}\coloneqq f_\psi + f_\phi$ and $f_{\psi\land\phi}\coloneqq f_\psi \cdot f_\phi$.

\item Consider $\psi(\tx)=\Ene y\,\phi(\tx,y)$ and assume w.l.o.g.\ that $y=x_{i+1}$.
For the inner formula, we have a polynomial
$f_\phi$ with indeterminates in $\XX{i+1}$ which we write as $f_\phi(\XX i,\YY i)$, where $\YY i=\XX{i+1}\setminus \XX i$.
Let $S$ be the set of all consistent selector functions $s\from \YY i\to \{0,1\}$, consistent in the sense that precisely one of $X_\alpha$, $X_{\neg\alpha}$ is mapped to 0, for all $\tau$-atoms $\alpha$.
Now set $f_\psi(\XX i)\coloneqq \sum_{s\in S} f_\phi(\XX i,s(\YY i))$.

\item 
Finally consider $\psi(\tx)=\Ane y\, \phi(\tx,y)$ with $y=x_{i+1}$ and again write $f_\phi(\XX i, \YY i)$ as above.
Let $S$ be the set of all consistent selector functions $s\from \YY i\to \{0,e\}$ and set $f_\psi(\XX i)\coloneqq \prod_{s\in S} f_\phi(\XX i,s(\YY i))$.
\end{itemize}
\end{Definition}

\begin{figure}
\begin{tabular}{l|l}
$\psi$ & $f_\psi$ \\ \hline
$x_i = x_j$ & $1$ or $0$ (depending on $i=j$) \\
$R\tz$, $\neg R\tz$ & $X_{R\tz}$, $X_{\neg R\tz}$ \\
$\phi \lor \theta$, $\phi \land \theta$ & $f_\phi + f_\theta$, $f_\phi \cdot f_\theta$ \\
$\Ene y\ \phi(\tx,y)$ & $\sum_{s \in S} f_\phi(\XX i, s(\YY i))$, with consistent assignments $s \from \YY i \to \{0,1\}$ \\
$\Ane y\ \phi(\tx,y)$ & $\prod_{s \in S} f_\phi(\XX i, s(\YY i))$, with consistent assignments $s \from \YY i \to \{0,e\}$ \\
\end{tabular}
\caption{Construction of the polynomial $f_\psi$ (\cref{def-polynomial}).}
\label{figPolynomials}
\end{figure}

\begin{Example}
\label{exPolynomials}
Consider
$\psi\coloneqq \Ene x(\neg Exx\land\Ane y(Exy\lor (\neg Exy\land \Ene z(Exz\land Ezy))))$,
an $\FO^3$-sentence defining the directed graphs that contain some centre from which all nodes are reachable in
one or two steps. For ease of notation we abbreviate the indeterminates associated with the atoms as
$X\coloneqq X_{Exx}$, $Y\coloneqq X_{Exy}$, $Z\coloneqq X_{Exz}$ and $U\coloneqq X_{Ezy}$, as well as
$\nnX, \nn{Y}, \nn{Z}$ and $\nn{U}$ associated with the corresponding negated atoms. 
The following table describes the polynomials $f_\phi$ for the subformulae of $\psi$. 

\begin{center}
\begin{tabular}{|c|c|}
\hline
$\phi$&$f_\phi$\\ \hline\hline
$Exz \land Ezy$&$ZU$\\
$\Ene z ( Exz \land Ezy)$& 1\\ 
$\neg Exy\land \Ene z(Exz\land Ezy)$& $\nn{Y}$\\
$Exy\lor (\neg Exy\land \Ene z(Exz\land Ezy))$& $Y+\nn{Y}$\\
$\Ane y(Exy\lor (\neg Exy\land \Ene z(Exz\land Ezy))))$&$e$\\
$\neg Exx\land\Ane y(Exy\lor (\neg Exy\land \Ene z(Exz\land Ezy))))$&$e \nnX$\\
$\psi$&$e$\\
\hline
\end{tabular}
\end{center}
We remark that the classically equivalent sentence $\psi'$ obtained from $\psi$
by omitting the literal $\neg Exy$ is instead described by $f_{\psi'}=1$. 
\end{Example}

We next observe that polynomials $f\in E[\XX k]$, with indeterminates $X_\beta$ for literals $\beta \in \Lit_k(\tau)$,
are evaluated to  semiring values $f[\rho] \in K$
by atomic $k$-types $\rho\from\Lit_k(\tau)\to K$,  for any semiring $K$ with a distinguished element $\epsilon$.
Indeed, $\rho$ defines a unique homomorphism $h^\epsilon_\rho \from E[\XX k] \to K$,
induced by $h^\epsilon_\rho(e)\coloneqq \epsilon$ and $h^\epsilon_\rho(X_\beta)\coloneqq \rho(\beta)$ for every literal $\beta \in \Lit_k(\tau)$.
We put $f[\rho] \coloneqq h_\rho^\epsilon(f)$ and remark that by monotonicity of polynomials over semirings, we have that
$f[\rho]\leq f[\rho']$ whenever $\rho\leq\rho'$.

\section{The 0-1 law for lattice semirings}\label{sect:finiteminmax}

We now use the polynomials $f_\psi$ to obtain a first 0-1 law for finite min-max semirings.
In fact, our result is slightly more general: we consider finite \emph{lattice semirings} $(K,\join,\meet,0,1)$ where the two operations are supremum and infimum with respect to a given partial order with least element $0$ and greatest element $1$.
Min-max semirings are then the special case where the order is linear.
Notice that every bounded distributive lattice is a lattice semiring.

In such semirings, we define $\epsilon_K \coloneqq \Inf \{ j \in K : j \neq 0 \}$ as the smallest positive element, if such an element exists (otherwise $\epsilon_K = 0$).
In finite min-max semirings, we always have $\epsilon_K > 0$.
We now prove that for $K$-interpretations into finite lattice semirings with the $k$-extension property,
the polynomials $f_\psi$ constructed in \cref{def-polynomial} provide a concise and
adequate description of any first-order formula $\psi(\tx) \in \FO^k$.

\begin{theorem} \label{ext-polynomials} 
Let $(K,\join,\meet,0,1)$ be a finite lattice semiring,
$\tau$ a finite relational vocabulary and $k\in\omega$.
Then, for every $K$-interpretation $\pi\from\Lit_A(\tau)\to K$  with 
the $k$-extension property, every formula 
$\psi(x_1,\dots,x_i)\in\FO^k(\tau)$ and every tuple $\ta=(a_1,\dots,a_i)$
of pairwise distinct elements of $A$, we have that
$\fmla \psi \ta = \polyrho\psi{\ta}$.
\end{theorem}

\begin{proof}
We proceed by induction on $\psi$. If $\psi$ is a literal, the claim is immediate from
the definition of $f_\psi$.
For $\psi=\phi\lor\theta$, we have
$\fmla \psi \ta = \fmla \phi \ta \join \fmla \theta \ta =
\polyrho\phi{\ta}\join \polyrho\theta{\ta} =
(f_\phi + f_\theta)[\rho^\pi_{\ta}]$ by induction.
Analogously for conjunctions.

\medskip
Let now $\psi(\tx)=\Ene y\,\phi(\tx,y)$ and w.l.o.g.\ $y=x_{i+1}$. Recall that
$f_\psi(\XX i)$ is defined as $\sum_{s\in S} f_\phi(\XX i,s(\YY i))$,
where $\YY i=\XX{i+1}\setminus \XX i$ and
$S$ is the set of consistent selector functions $s\from \YY i \to \{0,1\}$.
Notice that when we evaluate $f_\psi$ in a lattice semiring, the sum is interpreted as supremum (and multiplication as infimum).
By induction,
\[
  \fmla \psi \ta =
  \Sup_{b\in A\setminus\ta} \fmla \phi {\ta,b} =
  \Sup_{b\in  A\setminus\ta} \polyrho\phi{\ta,b}.
  \tag{$*$}
\]

We first prove that $\polyrho \psi \ta$ is an upper bound for $\fmla \psi \ta$.
For every $b \in A \setminus \ta$, define the selector function $s_b$ by
$s_b(X_\beta) = 1$ if $\rhopi_{\ta,b}(\beta) \neq 0$ (and $s_b(X_\beta)=0$ otherwise), for every literal $\beta\in\Lit_{i+1}(\tau)\setminus\Lit_i(\tau)$.
Since $1$ is the largest semiring value, we have\footnote{We kindly ask the reader to permit the abbreviation $\poly \phi {\rhopi_\ta, s_b(\YY i)}$ of the technically correct, but more verbose $f_\psi(\XX i, s_b(\YY i))[\rho_\ta]$.}
$\polyrho \phi {\ta,b} \le \poly \phi {\rhopi_\ta, s_b(\YY i)}$ by monotonicity.
Hence $\fmla \psi \ta \le \Sup_{s \in S} \poly \psi {\rhopi_\ta,s(\YY i)} = \polyrho \psi \ta$ by $(*)$.

The other direction holds by the extension property.
Every selector function $s \in S$ induces an extension $\rho_s \in \extend(\rho^\pi_{\ta})$
with $\rho_s(\beta) = s(X_\beta)$ for $\beta \in \Lit_{i+1}(\tau) \setminus \Lit_i(\tau)$.
Since $\pi$ has the $k$-extension property, there is $b_s \in A \setminus \ta$ with $\rho^\pi_{\ta, b_s} = \rho_s$,
hence $\poly \phi {\rhopi_\ta, s(\YY i)} = \polyrho \phi {\ta,b_s}$.
As this holds for all $s$, we have $\polyrho \psi \ta \le \fmla \psi \ta$ by $(*)$ and thus equality.

\medskip
Finally let $\psi(\tx)=\Ane y\, \phi(\tx,y)$ and recall that $f_\psi(\XX i)$ is defined as $\prod_{s\in S} f_\phi(\XX i,s(\YY i))$,
where this time we consider selector functions $s \from \YY i \to \{0,e\}$ instead of $\{0,1\}$.
We again have $\fmla \psi \ta = \Inf_{b\in  A\setminus\ta} \polyrho\phi{\ta,b}$ by induction.
Since $\eps_K$ is the smallest positive semiring value (or $0$), we first observe that, completely analogous to the previous case,
$\polyrho \psi \ta$ is a lower bound for $\fmla \psi \ta$.
If $\eps_K > 0$, then the other direction is analogous as well:
for each $s \in S$, define $\rho_s$ by $\rho_s(\beta) = \eps_K$ if $s(X_\beta)=e$, and $\rho_s(\beta) = 0$ if $s(X_\beta)=0$ (recall that $e$ becomes $\eps_K$ when evaluating $f_\psi$);
this extension is realised by the extension property.

It remains to prove $\polyrho \psi \ta \ge \fmla \psi \ta$ in the case $\eps_K = 0$
(defining $\rho_s$ by setting $\rho_s(\beta)=\eps_k$ or $\rho_s(\beta)=0$ would not be consistent).
Recall that min-max semirings always have $\eps_K > 0$, so this case only happens for lattice semirings where the underlying order is not total.
Let $\min(K)$ be the set of minimal non-zero elements of $K$.
Observe that $|\min(K)| \ge 2$ and $\Inf \min(K) = 0$, as $K$ is finite and $\eps_K = 0$.
Let $R$ be the set of extensions $\rho \in \extend(\rhopi_\ta)$ such that $\rho(\beta)=0$ or $\rho(\beta) \in \min(K)$,
for all $\beta \in \Lit_{i+1}(\tau) \setminus \Lit_i(\tau)$.
By the $k$-extension property, all $\rho \in R$ are realised by some $b \in A \setminus \ta$,
hence $\Inf_{\rho \in R} \poly \phi \rho \ge \Inf_{b\in  A\setminus\ta} \polyrho\phi{\ta,b} = \fmla \psi \ta$.

Now consider $\pol \psi$.
As we evaluate $e$ to $\eps_K=0$, the selector function $s$ does not matter and we have $\polyrho \psi \ta = \poly \phi {\rhopi_\ta,0}$
(that is, we map all variables $X_\alpha,X_{\neg\alpha} \in \YY i$ to $0$, ignoring the usual consistency requirement).
We claim that $\poly \phi {\rhopi_\ta,0} = \Inf_{\rho \in R} \poly \phi \rho$.
To see this, we write $f_\phi = g + h$ or, more precisely, $f_\phi(\XX i, \YY i) = g(\XX i, \YY i) + h(\XX i)$,
where $g$ contains all the monomials of $f_\phi$ that contain any $X_\beta \in \YY i$, and $h$ the remaining ones.
Recall that when we evaluate $f_\psi = g + h$, we interpret addition by the semiring operation $\join$.
Since lattice semirings are distributive and $R$ finite, we have
\[
  \Inf_{\rho \in R} \poly\phi\rho =
  \Inf_{\rho \in R} (g[\rho] \join h[\rho^\pi_\ta]) =
  \big(\Inf_{\rho \in R} g[\rho]\big) \,\join\, h[\rhopi_\ta].
\]
Now consider any minimal element $\bot \in \min(K)$ and some type $\rho \in R$ with $\rho(\beta) \in \{0,\bot\}$ for all $\beta \in \Lit_{i+1}(\tau) \setminus \Lit_i(\tau)$.
By definition, each monomial $m$ of $g$ contains an indeterminate $X_\beta$ for some $\beta \in \Lit_{i+1}(\tau) \setminus \Lit_i(\tau)$,
so $m[\rho] \le \bot$ (recall that multiplication is $\meet$).
Hence $g[\rho] \le \bot$.
As this holds for all $\bot \in \min(K)$, we have shown $\Inf_{\rho \in R} g[\rho] \le \Inf \min(K) = 0$.
It follows that $\Inf_{\rho \in R} \poly \phi \rho = h[\rhopi_\ta] = \poly \phi {\rhopi_\ta,0}$ as claimed.
\end{proof}

\begin{corollary}[0-1 law for $\FO$ on finite lattice semirings] \label{0-1-law-finite}
Let $K$ be a finite lattice semiring, with a 
probability distribution $p\from K^+ \to (0,1]$, and let $\tau$ be a relational vocabulary.
Then, for every sentence $\psi\in\FO(\tau)$ and every value $j\in K$,
the sequence $(\mu_{n,p}[ \pi\ext{\psi}=j])_{n<\omega}$ converges exponentially fast to either 0 or 1, as $n$ goes to infinity.
Further, the only possible almost sure valuations of sentences are $\as_{K,p}(\FO(\tau))=\{0,1,\epsilon_K\}$.
\end{corollary}

\begin{proof}
Fix $k$ such that $\psi\in\FO^k(\tau)$. By \cref{extension-property} the 
probability that a random $K$-interpretation $\pi$ on $[n]$ has the $k$-extension property
converges to 1 exponentially fast, as $n$ goes to infinity. But on $K$-interpretations with the $k$-extension property,
$\psi$ is described by a polynomial $f_\psi$.
Since $\psi$ has no free variables, we have that either $f_\psi=0$, or $f_\psi=1$, or $f_\psi=e$,
and the atomic type to consider is the trivial empty type $\emptyset$, 
which implies that $\poly\psi\emptyset$ is either 0, or 1, or $\epsilon_K$.
By applying \cref{ext-polynomials}, we conclude that the probabilities  $\mu_{n,p}\smash{\big[} \pi\ext{\psi} = \poly\psi\emptyset \smash{\big]}$ converge to 1 exponentially fast.
\end{proof}

Notice that, as in the Boolean case, the 0-1 law does not extend to arbitrary formulae with free variables.
Indeed for an atomic formula, say $Exy$, any value $j \in K^+$ and any fixed pair of constants $k,\ell\in\N$,
we have that $\lim_{n\to\infty}\mu_{n,p}[ E(k,\ell) = j ] = \frac 1 2 p(j)$, which is in general not 0 or 1. 
Nevertheless we can extend the 0-1 law to formulae $\psi(\tx)$ with free variables, with the additional
constraint that every relational atom contains a quantified variable; this implies that $f_\psi$ is either $0$, $1$ or $e$.

\begin{corollary} \label{0-1-law-formulae}
Let $K$, $p$, $\tau$ be as in \cref{0-1-law-finite}.
Let $\Phi$ be the set of fully instantiated first-order sentences $\psi(\ta)$
where $\psi(x_1,\dots,x_i)$ is a formula in $\FO(\tau)$ with free variables $x_1,\dots,x_i$,
in which every relational atom contains a quantified variable, and
$\ta=(a_1,\dots,a_i)$ is a tuple of distinct natural numbers, i.e.\ of elements of 
all universes $[n]$ for large enough $n$. 
Then the 0-1 law holds for $K,p$, and $\Phi$, and $\as_{K,p}(\Phi)=\{0,1,\epsilon_K\}$.
\end{corollary}

\Cref{0-1-law-finite} splits the relational first-order sentences into three classes, according to
whether their valuations in finite lattice semirings are almost surely 0,1, or $\epsilon_K$.
Notice that this split is the same for all finite lattice semirings, since it just depends
on the associated polynomial $f_\psi$.
The only lattice semiring with two elements is the Boolean semiring (where we have $\eps_K = 1$).
The classical 0-1 law for first order logic, saying that every relational first-order sentence is 
asymptotically either almost surely true, or almost surely false, can thus be seen as a special case of \cref{0-1-law-finite}.
In particular, the almost sure valuations $\eps_K$ and $1$ in any finite lattice semiring $K$ occur precisely for the formulae which are almost surely true in the Boolean case.

\begin{Example}[secret facts]
Semiring semantics can be used to model access restrictions to atomic facts,
for reasoning about the necessary clearance level for checking the truth of logical statements. Specifically, the
\emph{access control semiring}, also called \emph{security semiring}, which has been studied for instance in \cite{FosterGreTan08} 
is a min-max semiring with elements $0 < \Tsec < \Sec < \Cnf < \Pub$ where 0 stands for ``inaccessible'' (or ``false''),
$\Tsec$ is ``top secret'', $\Sec$ is ``secret'', $\Cnf$ is ``confidential'', and
$\Pub$ is ``public''. An interpretation $\pi$ into this semiring labels atomic facts by access restrictions and
the associated valuation $\pi\ext\phi$ of a first-order statement $\phi$ describes the clearance level that
is necessary to verify the truth of $\phi$ under these restrictions.
\Cref{0-1-law-finite} implies that under a random assignment of access restrictions
(assuming positive probabilities of all security levels), any first-order statement can
almost surely either be checked with publicly available information, cannot be checked at all,
or requires clearance for top secret information.
\end{Example}

\section{Complexity}

We now study the complexity of computing the almost sure valuation of a given 
first-order sentence $\psi$ in finite lattice semirings. As shown above, this amounts
to the computation of the associated polynomial $f_\psi$.
While, for a sentence $\psi$, the polynomial $f_\psi$ is either 0, 1, or $e$,
the polynomials $f_\phi(\XX k)$ associated with formulae $\phi(x_1,\dots,x_k)$ 
are much more complicated and can have exponential length.  
Rather than computing these intermediate polynomials explicitly, we shall
present a recursive procedure for computing the values
$f_\phi[\rho]$ for any formula $\phi(x_1,\dots,x_k)\in\FO(\tau)$ 
and any atomic $k$-type $\rho\from\Lit_k(\tau)\to K$ with values in a finite min-max semiring.

We remark that the polynomial $f_\phi$  is the same for all finite lattice semirings.
For determining the almost sure valuations of first-order sentences it would 
therefore suffice
to define the procedure just for the three-element semiring $E$.
However, we can solve, with moderate additional effort, the more general
problem of computing valuations $\pi\ext{\psi(\tup a)}$ 
of formulae with free variables not just for $E$, but for any finite min-max semiring $K$,
and any $K$-interpretation $\pi\from\Lit_A(\tau)\to K$ with the $k$-extension property.
Indeed, by  \cref{ext-polynomials} we know that $\fmla \psi \ta = \polyrho\psi{\ta}$.

\medskip
We first prove that this evaluation problem can be solved in \pspace,
for any finite min-max semiring $K$.
Using the well-known fact that $\pspace$ coincides with alternating polynomial time,
we present the evaluation algorithm as an alternating procedure
$\textbf{Eval}(\psi,\rho,c)$ which, given  $\psi(x_1,\dots,x_k)\in\FO(\tau)$, 
an atomic $k$-type $\rho\from\Lit_k(\tau)\to K$, and a value $c\in K$
determines whether $f_\psi[\rho] =c$ (avoiding an explicit construction of $f_\psi$).
We assume that the reader is familiar with the notion of an alternating algorithm and its presentation 
as a game between an existential and a universal player (see e.g. \cite{BalcazarDiaGab90}).

For a complexity analysis, it is appropriate to assume
that formulae are written with the standard quantifiers $\E $ and $\A $,
rather than $\Ene$ and $\Ane$, since the elimination of standard quantifiers  
by excluding ones can increase the length of formulae exponentially. As a consequence, 
when treating quantifiers, the evaluation procedure will have to deal with potential equalities 
between different variables. Accordingly, for a formula $\phi=\E x_{k+1}\theta(x_1,\dots,\allowbreak x_k,\allowbreak x_{k+1})$
we have the polynomial $f_\phi \coloneqq \sum_{i=1}^k f_{\theta(x_1,\dots x_k,x_i)} + f_{\Ene x_{k+1}\theta}$
and analogously for universal quantifiers.

\medskip The idea of the evaluation procedure is that, at any step where it has to be verified 
whether $f_\phi[\rho] =c$ for some triple $(\phi,\rho,c)$, the existential player guesses 
values $c_i$ for the
immediate subformulae $\phi_i$ of $\phi$ which, if correct, would imply that 
indeed $f_\phi[\rho] =c$. The universal player then challenges one of
these claims. For formulae of the form $\E x_{i+1}\theta$ or $\A x_{i+1}\theta$, this involves
(existential and/or universal) choices of selector functions $s\from\YY i\to \{0,1\}$
or  $s\from\YY i\to \{0,e\}$ and the modification of $\rho\from\Lit_i(\tau)\to K$ to 
the extended type $\rho s\from\Lit_{i+1}(\tau)\to K$ defined by
\[
    (\rho s)(\alpha)=\begin{cases}
        \rho(\beta) &\text{ if $\beta\in\Lit_i(\tau)$,}\\
        s(X_\beta)  &\text{ if $\beta\in\Lit_{i+1}(\tau)\setminus \Lit_i(\tau)$.}
    \end{cases}
\]
The procedure ends at triples $(\phi,\rho,c)$ where $\phi$ is atomic, at which
point the algorithm just checks whether $\rho(\phi)=c$.
A detailed description of the algorithm for any relational vocabulary $\tau$ and any min-max semiring $(K,\max,\min,0,1)$ is given in \cref{figAlgorithm}.

\begin{figure}[t]
\begin{mdframed}
\begin{multicols}{2}
\small
\begin{tabbing}
\phantom{x} \= \phantom{x} \= \phantom{x} \= \kill
\textbf{Eval}($\psi,\rho,c$), \textbf{input}:\\
\> a formula $\psi(x_1,\dots,x_k)\in\FO(\tau)$ in nnf\\
\> an atomic type $\rho\from\Lit_k(\tau)\to K$\\
\> an element $c\in K$\\[.5em]
\IF $\psi$ is an atom or negated atom \THEN\\
\> \ACCEPT \IF $\rho(\phi)=c$, \ELSE \REJECT \\[.5em]
\IF $\psi=\phi_1\lor\phi_2$ \THEN  \\
\> \GUESS $c_1,c_2\in K$ with $\max(c_1,c_2)=c$\\
\>  \CHOOSE $i\in\{1,2\}$\\
\>  \textbf{Eval}($\phi_i,\rho,c_i$) \\[.5em]
\IF $\psi=\phi_1\land\phi_2$ \THEN  \\
\> \GUESS $c_1,c_2\in K$ with $\min(c_1,c_2)=c$\\
\>  \CHOOSE $i\in\{1,2\}$\\
\>  \textbf{Eval}($\phi_i,\rho,c_i$) \\[.5em]
\IF $\psi=\E x_{k+1}\phi$ \THEN \\ 
\> \GUESS $c_1,\dots c_{k+1}$ s.t.\ $\max(c_1,\dots,c_{k+1})=c$\\ 
\>  \CHOOSE $i\in\{1,\dots,k+1\}$
\end{tabbing}
\begin{tabbing}
\phantom{x} \= \phantom{x} \= \phantom{x} \= \kill
\>  \IF $i\leq k$ \THEN \\
\>\> set $\theta(x_1,\dots,x_k)\coloneqq\phi(x_1,\dots,x_k,x_i)$ \\
\>\> \textbf{Eval}($\theta,\rho,c_i$)\\
\> \IF $i=k+1$ \THEN \\
\>\> \GUESS $s\from\YY k\to\{0,1\}$\\
\>\> \CHOOSE $s'\from\YY k\to\{0,1\}$\\
\>\> \IF $s'=s$ \THEN \textbf{Eval}($\phi,\rho s, c_{k+1}$)\\
\>\> \ELSE \GUESS  $c'\leq c_{k+1}$ and \textbf{Eval}($\phi,\rho s', c'$)\\[.5em]
\IF $\psi=\A x\phi$ \THEN \\
\> \GUESS $c_1,\dots c_{k+1}$ s.t.\ $\min(c_1,\dots,c_{k+1})=c$\\ 
\>  \CHOOSE $i\in\{1,\dots,k+1\}$\\
\>  \IF $i\leq k$ \THEN \\
\>\> set $\theta(x_1,\dots,x_k)\coloneqq\phi(x_1,\dots,x_k,x_i)$ \\
\>\> \textbf{Eval}($\theta,\rho,c_i$)\\
\> \IF $i=k+1$ \THEN \DO\\
\>\> \GUESS $s\from\YY k\to\{0,\epsilon_K\}$\\
\>\> \CHOOSE $s'\from\YY k\to\{0,\epsilon_K\}$\\
\>\> \IF $s'=s$ \THEN \textbf{Eval}($\phi,\rho s, c_{k+1}$)\\
\>\> \ELSE \GUESS  $c'\geq c_{k+1}$ and \textbf{Eval}($\phi,\rho s', c'$)
\end{tabbing}
\end{multicols}
\end{mdframed}
\vspace{-1em} 
\caption{Alternating procedure \textbf{Eval}($\psi$,$\rho$,$c$) to decide $f_\psi[\rho] = c$ in min-max semirings.}
\label{figAlgorithm}
\end{figure}

It is obvious that the algorithm runs in alternating polynomial time, but
it remains to prove that it is correct; we proceed by induction on $\psi$. 
Given a triple $(\psi,\rho,c)$ such that, indeed,
$f_\psi[\rho]=c$, the algorithm accepts by making the following existential choices. 
At a disjunction or conjunction, the existential player guesses the correct values
of the immediate subformula. For a formula $\E x_{k+1}\theta(x_1,\dots,x_k,x_{k+1})$
the existential player guesses the values $c_i=f_{\theta(x_1,\dots,x_k,x_i)}[\rho]$
and $c_{k+1}=f_{\Ene x_{k+1}\theta}[\rho]$. If the universal player challenges the value 
for some $i\leq k$, the existential player wins the remaining game from the triple
$(\phi(x_1,\dots,x_k,x_i),\rho,c_i)$ by induction hypothesis.
If instead $c_{k+1}$ is challenged, then the existential player guesses some
selector function  $s\from\YY k\to\{0,1\}$
such that $c_{k+1}=f_{\Ene x_{k+1}\theta}[\rho]=f_\theta[\rho s]$.
The universal player challenges this by choosing also a function 
$s'\from\YY k\to\{0,1\}$. If $s'=s$ this corresponds to the challenge to prove that,
indeed, $f_\theta[\rho s]=c_{k+1}$; since this is the case, and by induction hypothesis, 
the existential player wins the remaining game. If $s'\neq s$ this corresponds to the
challenge to prove that $f_\theta[\rho s']\leq c_{k+1}$. The existential player answers this
by guessing the correct value $c' \coloneqq f_\theta[\rho s']$ and, again by the hypothesis, wins
the remaining game. For formulae $\A x_{k+1}\theta(x_1,\dots,x_k,x_{k+1})$,
the reasoning is analogous.

Consider now a triple $(\phi,\rho,c)$ such that $f_\psi[\rho]\neq c$. Then the existential player
must make incorrect guesses, and the universal player can make sure that such incorrect
triples are propagated through the play, and are then detected at the end, when an atomic formula
is evaluated. Consider again the case of a formula  $\phi=\E x_{k+1}\theta(x_1,\dots,x_k,x_{k+1})$.
From an incorrect triple $(\phi,\rho,c)$, the existential player guesses $c_1,\dots,c_{k+1}$
with $\max(c_1,\dots,c_{k+1})=c$. Hence either $(\theta(x_1,\dots,x_k,x_i),\rho,c_i)$ is incorrect
for some $i\leq k$, in which case the universal players chooses such an $i$ and wins by induction hypothesis,
or the triple $(\Ene x_{k+1}\theta,\rho,c_{k+1})$ is incorrect. In that case, for any function 
$s\from\YY k\to\{0,1\}$ that the existential player might guess, it is either the case that $f_\theta[\rho s]\neq c_{k+1}$,
in which case the universal players wins by choosing $s'=s$,
or that there exists another function $s'\from\YY k\to\{0,1\}$ with the property that 
$f_\theta[\rho s'] > c_{k+1}$. Whatever element $c'\leq c_{k+1}$ the existential player then
guesses, the universal player will then win the remaining game from the
incorrect triple $(\theta,\rho s',c')$. Again, the reasoning for universally quantified formulae is
completely analogous.

We thus have established the following result, for any finite min-max semiring $K$ 
and any relational vocabulary $\tau$.

\begin{theorem}
Given a formula $\psi(x_1,\dots,x_k)\in\FO(\tau)$ and
an atomic $k$-type $\rho\from\Lit_k(\tau)\to K$ in a fixed finite min-max semiring $K$, the value  $f_\psi[\rho]$ can be
computed in $\pspace$.
\end{theorem}

If we are only interested in the case where $\psi$ is a sentence, we can 
work over the min-max semiring $E$ and thus determine in $\pspace$ whether $f_\psi$ is
0,1, or $e$. On the other side,  is has been proved by Grandjean \cite{Grandjean83}
that, in classical Boolean semantics, the problem whether a given first-order sentence is almost surely true or 
almost surely false is \pspace-complete.

\begin{corollary} For any finite lattice semiring $K$, verifying the almost sure valuation
of first-order sentences in $K$ is  a $\pspace$-complete problem. 
\end{corollary} 

Grandjean's result readily implies that, for any semiring $K$, deciding whether or not 
the almost sure $K$-valuation of a first-order sentence is $0$, is  $\pspace$-complete as well.
However, it might still be the case that if it is known that $\psi$ is almost surely true
in the Boolean sense, then the problem whether its almost sure valuation in a finite lattice semiring is
1 or $\epsilon_K$ could be solved more efficiently. However, this is not the case.

\begin{theorem} The problem to decide whether a given almost surely true first order sentence
evaluates in lattice semirings with at least three elements almost surely to 1, or to $\epsilon_K$,
is \pspace-complete.
\end{theorem}

\begin{proof} It remains to show \pspace-hardness. For any fixed finite structure $\AA$
with at least two elements, the problem of evaluating a given first-order sentence on $\AA$ is
\pspace-complete. In particular this holds if $\AA$ is just a two-element set without any relations,
i.e.\  $\AA=\{0,1\}$. Given a sentence $\psi\in\FO(\emptyset)$, we consider
$\psi^* \coloneqq \E 0\E 1(0\neq 1\land \psi')$ where $\psi'$ is obtained by
relativising all quantifiers to $\{0,1\}$, i.e.\ by replacing subformulae $\E x\phi$ by $\E x((x=0\lor x=1)\land \phi)$
and  $\A x\phi$ by $\A x((x=0\lor x=1)\ra \phi)$. Clearly if $\{0,1\}\models \psi$ then $\psi^*$
almost surely evaluates to 1 (on any semiring) and if  $\{0,1\}\not\models \psi$ then $\psi^*$
almost surely evaluates to 0.

Let now $P$ be a unary relation symbol and consider the reduction that maps 
any sentence $\psi\in\FO(\emptyset)$ to $\psi^*\lor \A x (Px\lor\neg Px) \in\FO(\{P\})$.
Notice that such a sentence is almost surely true in the Boolean sense, and that
the almost sure valuation of $ \A x (Px\lor\neg Px)$ is $\epsilon$ in any finite lattice semiring.
Hence  the almost sure valuation of $\psi^*\lor \A x (Px\lor\neg Px)$ is
1 if $\{0,1\}\models\psi$, and  $\epsilon$, otherwise.
This proves that deciding whether an almost surely true sentence evaluates to 1 or to $\epsilon$
in a lattice semiring with at least three elements is \pspace-hard.
\end{proof}

\section{The 0-1 law for infinite lattice semirings}\label{sec:infinite}

We now move to infinite lattice semirings $(K,\join,\meet,0,1)$, in particular to semirings defined over the real numbers.
In the case that $K$ is countable, we can define probability measures on $K$-interpretations as in \cref{sect:random}.
In the general case, we assume that we a have a 
probability space $(K^+,\mathcal F,p)$ whose underlying $\sigma$-algebra $\mathcal F$ contains
all intervals $[a,b] = \{ x \in K \mid a \le x \le b \}$ for $a,b \in K$ (notice that $[a,b]$ is a sublattice).
We thus get probabilities $p[x \in J]$ for all closed, open, and half-open intervals $J\subseteq K^+$.
We further assume that $p[x = 1] > 0$, i.e.\ we have a
positive probability that a randomly chosen value coincides precisely with the maximal semiring value.%
\footnote{This is a natural assumption in our context of random semiring interpretations, but it is not really
essential; large values can instead be treated in an analogous way as
we do for small positive ones.}

The measures $\mu_{n,p}$ for random $K$-interpretations with universe $[n]=\{0,\dots,n-1\}$ are induced by $p$ as in the finite case:
Again, we consider the probabilistic process which, for
each instantiated atom $R\ta$ over $[n]$ first makes a random choice whether $R\ta$ or $\neg R\ta$ is
true, each with probability $1/2$ (this is an arbitrary choice, any fixed probability would work),
and then assigns to the true literal a positive semiring value according to $p$,
so that we have a probability that $\pi(R\ta)\in J$ for every interval\footnote{More precisely: if $0 \in J$, then $\mu_{n,p}[\pi(R\ta) \in J] = \frac 1 2 + \frac 1 2 p[x \in J \without 0]$, otherwise $\mu_{n,p}[\pi(R\ta) \in J] = \frac 1 2 p[x \in J \without 0]$.} $J\subseteq K$.
We consider three cases concerning the probabilities of small positive semiring values.

\begin{Definition}\label{def:smallpositivevalues}
We say that the probability measure $p$ is \emph{$\epsilon$-bounded} on small semiring values, for $\epsilon \in K$, if one of the following cases applies.
\smallskip
\begin{enumerate}
\item \emph{$p$ is weakly $\eps$-bounded} if $p[x = \eps] > 0$ and $p[0<x\leq \delta]=0$ for all $\delta\in K^+$ with $\epsilon\not\leq \delta$.\\
In particular, the smallest possible positive value of a literal is $\epsilon$.

\smallskip
\item \emph{$p$ is strictly $\eps$-bounded} if $p$ is not weakly $\eps'$-bounded (for any $\eps'$) and further $p[0 < x \le \eps] = 0$ and $p[0 < x \leq \delta] > 0$ for all $\delta > \eps$.\\
That is, $p$ only admits positive values greater than $\epsilon$.
We include the case $\eps = 0$.
\end{enumerate}
\smallskip
\noindent To avoid going through case distinctions in the proofs to follow, we say that
a semiring value $\delta\in K^+$ is \emph{$p$-relevant}, if 
either $\delta>\epsilon$, or if $p[x=\epsilon]>0$ and $\delta=\epsilon$.
Moreover, we write $\eps \ll \delta$ if there is a $\gamma \in K$ with $\eps < \gamma < \delta$.
\end{Definition}

In the remainder of this section we consider infinite lattice semirings  $(K,\join,\meet,0,1)$
together with a probability measure $p$ on $K^+$ assigning probabilities to all intervals, such that $p$ is $\epsilon$-bounded with $\eps \ll 1$.
We remark that we make the assumption $\eps \ll 1$ only to simplify the presentation, but this is not an actual restriction (one can easily verify that \cref{infinite-0-1,infinite-asv} also holds in the few special cases with $\eps \centernot\ll 1$).

%
%

\begin{Definition}
Let $\Phi \subseteq \FO(\tau)$.
We say that a \emph{0-1 law} holds for $\Phi$, $K$, and $p$
if for each sentence $\psi\in\Phi$ and each interval $J\subseteq K$
the sequence $(\mu_{n,p}[\pi\ext{ \psi}\in J])_{n<\omega}$ converges to either 0 or 1, as $n$ goes to infinity.

We further say that $j$ is the \emph{almost sure valuation} of $\psi$ (for $K$ and $p$), denoted $\as_{K,p}(\psi)=j$,
if there is a decreasing sequence $(J_i)_{i<\omega}$ of intervals $J_i \subseteq K$ with $\bigcap_{i<\omega} J_i =\{j\}$ such that  
$\lim_{n\to\infty} \mu_{n,p}[\pi\ext{\psi}\in J_i ]= 1$ for all $i<\omega$.
\end{Definition} 

We also have to define the extension properties a bit differently, as we cannot realise all possible extensions over an infinite semiring in a finite structure.

\begin{Definition}
Given an atomic $m$-type $\rho$, we say that $\rho^+\in\extend(\rho)$ is a \emph{maximal extension}
of $\rho$ if  $\rho^+(\beta)\in\{0,1\}$, for every literal $\beta\in\Lit_{m+1}(\tau)\setminus\Lit_m(\tau)$. 
Further we say that  $\rho^-\in\extend(\rho)$ is a \emph{$\delta$-small extension} of $\rho$, 
if $\rho^-(\beta)\leq \delta$ for every $\beta\in\Lit_{m+1}(\tau)\setminus\Lit_m(\tau)$.
\end{Definition}

We remark that, by definition of atomic types, a $\delta$-small extension $\rho^-$ maps out of each pair $\alpha,\neg\alpha$ of complementary literals that contain the variable $x_{m+1}$ precisely one to $0$ and the other one into the interval $(0,\delta]$.

\begin{Definition} A semiring interpretation $\pi\from\Lit_A(\tau)\to K$  
has the \emph{$(k,\delta)$-extension property}, where $\delta \in K^+$, if for every $m< k$, every tuple 
$\ta\in A^m$, and every maximal extension $\rho^+\in\extend(\rho^\pi_{\ta})$,
there exists
\begin{enumerate}
\item an element $b\in A \setminus \ta$ such that $\rho^\pi_{\ta,b}=\rho^+$, and
\item an element $c\in A \setminus \ta$ such that 
$\rho^\pi_{\ta,c}\leq\rho^+$ and $\rho^\pi_{\ta,c}$ is a $\delta$-small extension of $\rho^\pi_{\ta}$.
\end{enumerate}
\end{Definition} 

In other words, if $\pi$ has the $(k,\delta)$-extension property then
every realisation of an atomic $m$-type in $\pi$ can be extended to realisations
of all its maximal extensions, but also to realisations of $\delta$-small extensions (with the same underlying Boolean types as the maximal extensions).
 
\begin{proposition} Let $K$ be an infinite lattice semiring with an $\epsilon$-bounded probability measure $p$.
For every fixed $k$, every finite relational vocabulary $\tau$, and every $p$-relevant $\delta$,
\[\lim_{n\to\infty} \mu_{n,p}[ \pi\text{ has the $(k,\delta)$-extension property }] = 1,\]
and the convergence to this limit is exponentially fast.
\end{proposition}

\begin{proof}
For a given probability measure $\mu_{n,p}$ we first calculate a bound for the probability 
that a given realisation $\ta$ of an atomic $m$-type $\rho$ (with $m<k$)
cannot be extended to a realisation $\ta, b$
of a given \emph{maximal extension} $\rho^+$ of $\rho$. This is  analogous to
the argument in \cref{extension-property}.
For any pair $\alpha,\neg\alpha$ of complementary literals in
$\Lit_{m+1}(\tau)\setminus\Lit_m(\tau)$, the probability that
randomly chosen values according to $p$ for $\alpha$ and $\neg\alpha$
are $1$ and $0$, as prescribed by $\rho^+$, is $p[x=1]/2$. There is a fixed number $q$ of 
pairs of such literals, so the probability that all chosen values coincide with
those required by $\rho^+$ is a fixed number $\gamma\coloneqq (p([x=1]/2)^q$.
It follows that
\[  
    \mu_{n,p}[
    \text{ some realisation of $\rho$ does not extend to a realisation of $\rho^+$}
    ] \leq n^m(1-\gamma)^{n-m}
\]
which for growing $n$ converges to 0 exponentially fast. 

Let us now consider extensions with small truth values. Fix $\rho$ and some
maximal extension $\rho^+\in\extend(\rho)$.
For each $p$-relevant $\delta$ there exists a number $g(\delta)>0$ such that $p[ 0 < x \leq \delta] = g(\delta)$.
Hence the probability that values for complementary literals $\alpha$ and $\neg\alpha$
with the variable $x_{m+1}$, chosen according to $p$, define a $\delta$-small extension $\rho^-\leq\rho^+$
is $\gamma\coloneqq (g(\delta)/2)^q$, and with precisely the same calculation as above, we conclude
that
\[
    \mu_{n,p}[
    \text{ some realisation of $\rho$ does not extend to a realisation of some $\delta$-small
 $\rho^-\leq\rho^+$}
    ]
\]
converges to 0 exponentially fast. 
\end{proof}

We again use the polynomials $f_\psi$ of \cref{def-polynomial} to represent formulae $\psi(\tx) \in \FO^k$.
However, the evaluation of these polynomials must be more flexible, taking into account different parameters for
small positive values. 
Specifically, given $\delta>0$ and an atomic  $k$-type $\rho$, we evaluate
a polynomial $f\in E[\XX k]$ 
to a semiring value $f^\delta[\rho]\in K$, via the homomorphism 
$h^\delta_\rho\from E[\XX k]\to K$
induced by $h^\delta_\rho(e)\coloneqq \delta$ and $h^\delta_\rho(X_\beta)\coloneqq \rho(\beta)$,
for literals $\beta \in \Lit_k(\tau)$.
We can now formulate an analogue of \cref{ext-polynomials}, requiring only a mild assumption on the lattice structure:

\begin{Definition}
A lattice semiring $(K,\join,\meet,0,1)$ is called \emph{0-1-irreducible} if
$a \meet b = 0$ implies $a=0$ or $b=0$ (no divisors of $0$) and
$a \join b = 1$ implies $a=1$ or $b=1$.
\end{Definition}

Notice that both properties are always satisfied in min-max semirings (as the natural order is total).

\begin{theorem}\label{delta-extension}
Let $(K, \join, \meet, 0, 1)$ be a (possibly infinite) lattice semiring without divisors of $0$.
Let $\delta > 0$ and let $\pi \from \Lit_A(\tau) \to K$ be a $K$-interpretation with the $(k,\delta)$-extension property.
Then, for every formula $\psi(x_1,\dots,x_i)\in\FO^k(\tau)$ and every tuple $\ta\in A^i$, either
\begin{itemize}
\item $f^\delta_\psi[\rho^\pi_{\ta}] = \fmla \psi \ta = 0$, or
\item $f^\delta_\psi[\rho^\pi_{\ta}],\fmla \psi \ta \neq 0$ and $f^\delta_\psi[\rho^\pi_{\ta}] \le \fmla \psi \ta \join \delta \le f^\delta_\psi[\rho^\pi_{\ta}] \join \delta$.
\end{itemize}
\end{theorem}

\begin{proof}
The proof is by induction over $\psi$ along the lines of the proof of \cref{ext-polynomials}.
For simplicity, we drop the annotation $\delta$ in $f^\delta_\psi$ and refer to the two cases in the theorem as $(0)$ and $(\delta)$.
For literals, we always have $\polyrho \psi \ta = \fmla \psi \ta$ and either $(0)$ or $(\delta)$ holds.

\medskip\noindent
For $\psi = \phi \lor \theta$, we have $\pol \psi = \pol \phi + \pol \theta$ and $\fmla \psi \ta = \fmla \phi \ta \join \fmla \theta \ta$.
Recall that we evaluate the addition in $f_\psi$ by the semiring operation $\join$.
We distinguish the cases whether $(0)$ or $(\delta)$ applies to $\phi$ and $\theta$.
If both satisfy $(0)$, then $(0)$ also holds for $\psi$.
If $(\delta)$ holds for $\phi$ and $(0)$ for $\theta$ (or vice versa), then $(\delta)$ also holds for $\psi$.
If $(\delta)$ holds for both, then it also holds for $\psi$, since clearly $\polyrho\psi\ta,\fmla\psi\ta \neq 0$ and
\[
  \polyrho \phi \ta \join \polyrho \theta \ta \le (\fmla \phi \ta \join \delta) \join (\fmla \theta \ta \join \delta) = \fmla \psi \ta \join \delta
\]
and similarly for the second inequality.

\medskip\noindent
For $\psi = \phi \land \theta$, we have $\pol \psi = \pol \phi \cdot \pol \theta$ and $\fmla \psi \ta = \fmla \phi \ta \meet \fmla \theta \ta$.
If $(0)$ holds for $\phi$ or $\theta$, then $(0)$ also holds for $\psi$.
If $(\delta)$ holds for both $\phi$ and $\theta$, first observe that $\polyrho \phi \ta \neq 0$ and $\polyrho \theta \ta \neq 0$ imply $\polyrho \psi \ta \neq 0$ (no divisors of $0$), and analogously also $\fmla\psi\ta \neq 0$.
Then $(\delta)$ holds for $\psi$, since by induction,
\[
  \polyrho\phi\ta \meet \polyrho\theta\ta \le
  (\fmla \phi \ta \join \delta) \meet (\fmla \theta \ta \join \delta) \le
  (\polyrho\phi\ta \join \delta) \meet (\polyrho\theta\ta \join \delta),
\]
and by distributivity\footnote{While semiring distributivity only implies $a \meet (b \join c) = (a \meet b) \join (a \meet c)$, in lattice settings this also implies the dual law $a \join (b \meet c) = (a \join b) \meet (a \join c)$ which we use here.},
\[
  (\fmla \phi \ta \join \delta) \meet (\fmla \theta \ta \join \delta) =
  (\fmla \phi \ta \meet \fmla \theta \ta) \join (\delta \join \delta) =
  \fmla \psi \ta \join \delta,
\]
and similarly for $\polyrho \psi \ta \join \delta$.

\medskip\noindent
Let now $\psi(\tx)=\Ene y\,\phi(\tx,y)$. Recall that
\[
   \fmla\psi\ta = \Sup_{b\in A\setminus\ta} \fmla\theta\ta \quad\text{and}\quad \pol\psi(\XX i) \coloneqq \Sup_{s\in S} \pol\phi(\XX i, s(\YY i)),
\]
where $\YY i=\XX {i+1} \setminus \XX i$ and $S$ is the set of all consistent selector functions $s\from \YY i \to \{0,1\}$.

We first prove that $\polyrho\psi\ta = \Sup_{b\in  A\setminus\ta} \polyrho\phi{\ta,b}$.
Recall that each selector function $s$ induces the maximal extension $\rho_s$ with $\rho_s(\beta) = s(\beta)$ for the new literals $\beta$.
By the $(k,\delta)$-extension property, there is an element $b$ with $\rho_{\ta,b} = \rho_s$ and we then have $\poly\phi{\rho_\ta^\pi, s(\YY i)} = \polyrho\phi{\ta,b}$.
Hence $\polyrho\psi\ta \le \Sup_{b\in  A\setminus\ta} \polyrho\phi{\ta,b}$.
Conversely, let $b \in A \setminus \ta$ and consider the type $\rho^\pi_{\ta,b}$.
Let $\rho^+$ be the maximal extension induced by $\rho^\pi_{\ta,b}$ (i.e., with the same underlying Boolean type).
By the $(k,\delta)$-extension property, there is an element $b^+$ with $\rho^+ = \rho^\pi_{\ta,b^+}$.
Then $\rho^\pi_{\ta,b} \le \rho^\pi_{\ta,b^+}$ and hence $\polyrho\phi{\ta,b} \le \polyrho\phi{\ta,b^+}$ by monotonicity.
Setting $s(\beta) = \rho_b^+(\beta)$, we have $\polyrho\phi{\ta,b^+} = \poly\phi{\rho^\pi_\ta, s(\YY i)}$ and hence
$\Sup_{b\in  A\setminus\ta} \polyrho\phi{\ta,b} \le \polyrho\psi\ta$.

To prove that either $(0)$ or $(\delta)$ holds for $\psi$, we proceed by case distinction for each $b$.
If $(0)$ holds for all $\phi(\ta,b)$, then $\fmla\psi\ta = 0$ and also $\Sup_{b\in A\setminus\ta} \polyrho\phi{\ta,b} = 0$, so $(0)$ holds for $\psi$.
Otherwise, there is at least one $b$ with $\polyrho\phi{\ta,b},\fmla\phi{\ta,b} \neq 0$, hence $\polyrho\psi\ta, \fmla\psi\ta \neq 0$ as well.
We ignore all $b$ for which $(0)$ holds, as they do not affect the supremum.
Then $(\delta)$ holds for $\psi$:
\begin{align*}
  \polyrho\psi\ta = \Sup_{b\in  A\setminus\ta} \polyrho\phi{\ta,b} &\le \Sup_{b\in  A\setminus\ta} (\fmla\phi{\ta,b} \join \delta) = \fmla\psi\ta \join \delta \\
  &\le \Sup_{b\in  A\setminus\ta} (\polyrho\phi{\ta,b} \join \delta) = \polyrho\psi\ta \join \delta.
\end{align*}

\medskip\noindent
Finally, let $\psi(\tx)=\Ane y\,\phi(\tx,y)$. Recall that
\[
   \fmla\psi\ta = \Inf_{b\in A\setminus\ta} \fmla\theta\ta \quad\text{and}\quad \pol\psi(\XX i) \coloneqq \Inf_{s\in S} \pol\phi(\XX i, s(\YY i)),
\]
where now we consider selector functions $s\from \YY i \to \{0,\delta\}$.

As for existential quantification, we first relate the selector functions $s$ to the elements $b \in A\setminus\ta$.
Since $\pi$ only guarantees $\delta$-small extensions, we relax the equality by $\delta$:
\[
  \Big(\Inf_{b\in A\setminus\ta} \polyrho\phi{\ta,b}\Big) \join \delta \;\ge\;
  \polyrho\psi\ta \;\ge\;
  \Inf_{b\in A\setminus\ta} \polyrho\phi{\ta,b}.
\]

The second inequality is easy: For each selector function $s$, consider the maximal extension $\rho_s^+$ of $\rho^\pi_\ta$ induced by $s$ (i.e., with the same underlying Boolean type).
By the $(k,\delta)$-extension property, there is an element $c$ such that $\rho^\pi_{\ta,c}$ is a $\delta$-small extension with $\rho^\pi_{\ta,c} \le \rho_s^+$.
Then also $\rho^\pi_{\ta,c}(\beta) \le s(X_\beta)$ for all new literals $\beta$ by definition of $s$ and $\delta$-small, hence $\poly\phi{\rho^\pi_\ta,s(\YY i)} \ge \polyrho\phi{\ta,c}$ by monotonicity and the inequality follows.

For the first inequality, we consider the monomials of $\pol\phi(\XX i, \YY i)$ and split the polynomial into $\pol\phi(\XX i, \YY i) = g(\XX i) + h(\XX i, \YY i)$, where $g$ contains precisely those monomials that contain no indeterminates in $\YY i$.
We clearly have $\polyrho\phi{\ta,b} \ge g[\ta]$.
Recall that the universe $A$ is finite, so we can apply distributivity and obtain:
\begin{align*}
  \Big(\Inf_{b\in A\setminus\ta} \polyrho\phi{\ta,b} \Big) \join \delta =
  \Inf_{b\in A\setminus\ta} ( \polyrho\phi{\ta,b} \join \delta ) \ge
  g[\ta] \join \delta \ge
  \Inf_{s \in S} ( g[\ta] \join h(\ta,s(\YY i)) ) =
  \polyrho\psi\ta.
\end{align*}
For the last inequality, we use the fact that $s(X_\beta) \in \{0,\delta\}$ for all $X_\beta \in \YY i$ and hence $m[\ta,s(\YY i)] \le \delta$ for all monomials $m$ of $h$ by construction.

To prove that $(0)$ or $(\delta)$ holds for $\psi$, we again proceed by case distinction for each $b$.
First assume that $(0)$ holds for some $\phi(\ta,b)$, so $\fmla\phi{\ta,b} = \polyrho\phi{\ta,b} = 0$.
Then also $\fmla\psi\ta = 0$ and it remains to prove $\polyrho\psi\ta = 0$.
We again consider the monomials of $\pol\phi(\XX i, \YY i)$ and split the polynomial into $\pol\phi(\XX i, \YY i) = g(\XX i) + h(\XX i, \YY i)$ as above.
By $\polyrho\phi{\ta,b} = 0$, we must have $g[\ta] = 0$ and $h[\ta,b] = 0$.
As there are no divisors of $0$, this means that every monomial in $h[\ta,b]$ must contain a literal $X_\beta \in \XX i \cup \YY i$ such that $\rho^\pi_{\ta,b}(\beta) = 0$.
Let $s$ be any selector function such that whenever $\rho^\pi_{\ta,b}(X_\beta) = 0$ for $X_\beta \in \YY i$, also $s(X_\beta) = 0$ (the other values can be chosen arbitrarily).
Observe that such a selector function exists in $S$ since $\rho^\pi_{\ta,b}$ is a type (i.e., consistent on opposing literals).
Then $h[\ta,b] = h[\ta,s(\YY i)] = 0$ by construction of $s$ and it follows that $\polyrho\psi\ta = 0$.

Lastly, assume that $(\delta)$ holds for all $b$.
Then $\polyrho\phi{\ta,b},\allowbreak \fmla\phi{\ta,b} \neq 0$ for all $b$ and thus $\polyrho\psi\ta,\allowbreak \fmla\psi\ta \neq 0$, since there are no divisors of $0$ (recall that the infimum is over a finite universe or set $S$).
Using the relaxed equality and distributivity, we obtain:
\begin{align*}
\polyrho\psi\ta \le
\Inf_{b\in A\setminus\ta} (\polyrho\phi{\ta,b} \join \delta) &\le
\Inf_{b\in A\setminus\ta} ((\fmla\phi{\ta,b} \join \delta) \join \delta) = \fmla\psi\ta \join \delta \\ &\le
\Inf_{b\in A\setminus\ta} ((\polyrho\phi{\ta,b} \join \delta) \join \delta) = \polyrho\psi{\ta,b} \join \delta. \qedhere
\end{align*}
\end{proof}

The above theorem essentially establishes a relaxed version of the equality
$f^\delta_\psi[\rho^\pi_{\ta}] = \fmla \psi \ta$ that holds in finite lattice semirings.
The reason is the ($k,\delta$)-extension property, which does not guarantee that the value $\delta$ is assumed by extensions, but only makes the weaker guarantee that some values in $(0,\delta]$ are assumed.
For sentences, the relaxed equality reduces to the following three cases.

\begin{corollary}\label{delta-extension-sentence}
Let $(K, \join, \meet, 0, 1)$ be a (possibly infinite) 0-1-irreducible lattice semiring.
Let $0 < \delta < 1$ and let $\pi \from \Lit_A(\tau) \to K$ be a $K$-interpretation with the $(k,\delta)$-extension property.
Then, for every sentence $\psi \in \FO(\tau)$,
\begin{itemize}
\item if $\pol\psi = 0$, then also $\fml\psi = 0$;
\item if $\pol\psi = 1$, then also $\fml\psi = 1$;
\item if $\pol\psi = e$, then $0 < \fml\psi \le \delta$.
\end{itemize}
\end{corollary}

\begin{proof}
The first statement is immediate by \cref{delta-extension}.
For the second statement, \cref{delta-extension} implies $\fml\psi \join \delta = 1$.
By assumption on $K$ and $\delta < 1$, this implies $\fml\psi = 1$.
For the last statement, recall that we have $\poly[\delta]\psi\emptyset = \delta$ for $\pol\psi = e$.
 \cref{delta-extension} states $\fml\psi \neq 0$ as well as $\fml\psi \join \delta = \delta$ which implies $\fml\psi \le \delta$.
\end{proof}

To determine which intervals occur almost surely in the case $f_\psi = e$, we need the following simple observation.

\begin{lemma}\label{lemma-directed-interval}
Let $J \subseteq K$ be a directed interval, i.e., $x,y \in J$ implies $x \meet y \in J$ and $x \join y \in J$.
If $\pi \from \Lit_A(\tau) \to J \cup \{0,1\}$ is a $K$-interpretation that maps all literals into $J$ (or to $0$ or $1$), then also $\pi \ext \psi \in J \cup \{0,1\}$, for every sentence $\psi \in \FO(\tau)$.
\end{lemma}

\begin{proof}
Straight-forward induction on $\psi$.
Recall that we assume the universe $A$ to be finite, so all logical operators are evaluated as finite $\meet$ or $\join$ and the value thus remains in $J \cup \{ 0,1 \}$.
\end{proof}

\begin{corollary}[0-1 law for $\FO$ on infinite lattice semirings]\label{infinite-0-1}
Let $(K,\join,\meet,0,1)$ be a 0-1-irreducible lattice semiring with $\epsilon$-bounded probability measure, where $\eps \ll 1$.
Then, for every sentence $\psi\in\FO(\tau)$ over relational vocabulary $\tau$ and every interval $J\subseteq K$,
the sequence $(\mu_{n,p}[ \pi\ext{\psi} \in J ])_{n<\omega}$ converges exponentially fast to either 0 or 1, as $n$ goes to infinity.

Further, the only intervals $J$ for which 
$\lim_{n\to\infty} \mu_{n,p}[ \pi\ext{\psi} \in J ] = 1$ is possible are those where
either $0\in J$, $1\in J$, $\epsilon\in J$, or $(\epsilon,\delta)\subseteq J$ for some $\delta>\epsilon$.
\end{corollary}

\begin{proof}

Fix $k$ such that $\psi\in\FO^k(\tau)$, and consider the associated polynomial $f_\psi$.
Since $\psi$ is a sentence, we have $f_\psi \in \{0,1,e\}$.
For every $p$-relevant $\delta$, the sequence
\[ \mu_{n,p}[ \pi\text{ has the $(k,\delta)$-extension property }]\]
converges to 1 exponentially fast.
We first consider the case that $f_\psi=0$ or $f_\psi=1$.
Since $\eps \ll 1$, there is a $p$-relevant $\delta < 1$ and 
\cref{delta-extension-sentence} thus implies that
$\lim_{n\to\infty} \mu_{n,p}[  \pi\ext{\psi} \in J ]= 1$ if $f_\psi\in J$ 
and $\lim_{n\to\infty} \mu_{n,p}[  \pi\ext{\psi} \in J ]= 0$ otherwise.

Now consider the case that $f_\psi=e$.
For every $p$-relevant $\delta$ and any $K$-interpretation $\pi$ with the $(k,\delta)$-extension property, \cref{delta-extension-sentence}
implies $0 <\pi\ext{\psi} \leq \delta$.
Since the $(k,\delta)$-extension property is asymptotically almost surely satisfied,
it follows that
\[ \lim_{n\to\infty}\mu_{n,p}[ 0< \pi\ext{\psi}\leq \delta ] =1 \tag{$*$}\]
with exponential convergence, for every $p$-relevant $\delta$.
We get back to the two cases concerning the parameter $\epsilon$ of $p$:

\begin{enumerate}
\item $p$ is weakly $\eps$-bounded: $p[x = \eps] > 0$ and $p[0<x\leq \delta]=0$ for all $\delta\in K^+$ with $\epsilon\not\leq \delta$.

\smallskip
Since $p$ only admits values in the closed (and hence directed) interval $J_\eps = [\eps,1]$, \cref{lemma-directed-interval} implies $\mu_{n,p}[ \pi \ext \psi \in J_\eps] = 1$ for all $n$.
Conversely, $\eps$ is $p$-relevant, so together with $(*)$,
\[
  \lim_{n\to\infty}\mu_{n,p}[ \pi\ext{\psi}\in J ] = \begin{cases}
    1&\text{ if $\epsilon\in J$},\\
    0&\text{ otherwise.}
  \end{cases}
\]

\item $p$ is strictly $\eps$-bounded: $p[0 < x \le \eps] = 0$ and $p[0 < x \leq \delta] > 0$ for all $\delta > \eps$.

\smallskip
First assume that there are $\delta,\delta' > \eps$ with $\delta \meet \delta' = \eps$ (which implies $\delta,\delta' < 1$).
If $\pi$ has both the $(k,\delta)$- and the $(k,\delta')$-extension property, then \cref{delta-extension-sentence} implies
$\pi \ext \psi \le \delta$ and $\pi \ext \psi \le \delta'$, so $\pi \ext \psi \le \delta \meet \delta' = \eps$.
Since $\delta$ and $\delta'$ are $p$-relevant, both extension properties almost surely hold.
We further have $\pi \ext \psi \ge \eps$ by \cref{lemma-directed-interval}, since $p[\eps \le x \le 1]=1$ and the interval $[\eps,1]$ is closed (and hence directed).
Combining both bounds yields
\[
  \lim_{n\to\infty}\mu_{n,p}[ \pi\ext{\psi}\in J ] = \begin{cases}
    1&\text{ if $\epsilon\in J$},\\
    0&\text{ otherwise.}
  \end{cases}
\]

If no such $\delta,\delta'$ exist, then the interval $J_\eps = (\eps,1]$ is directed.
Since $p$ only admits values in $J_\eps$, \cref{lemma-directed-interval} implies $\mu_{n,p}[ \pi \ext \psi \in J_\eps] = 1$ for all $n$.
Together with $(*)$, we get
\[
  \lim_{n\to\infty}\mu_{n,p}[ \pi\ext{\psi}\in J ] = \begin{cases}
    1&\text{ if $(\eps,\delta) \subseteq J$ for some $\delta > \eps$}\\
    0&\text{ otherwise.}
  \end{cases}
\]
We remark that the intervals $(\eps,\delta)$ are non-empty (see the proof of \cref{infinite-asv}). \qedhere
\end{enumerate}
\end{proof}

\begin{corollary}[Almost sure valuations]\label{infinite-asv}
For every infinite lattice semiring $K$ with $\epsilon$-bounded probability measure $p$, where $\eps \ll 1$, and every relational vocabulary $\tau$,
the only possible almost sure valuations are
$\as_{K,p}(\FO(\tau)) = \{0,1,\epsilon\}$.
\end{corollary}

\begin{proof}
In the cases where $\lim_{n\to\infty} \mu_{n,p}[ \pi\ext{\psi} \in J ] = 1$ holds whenever $0 \in J$, $1 \in J$, or $\eps \in J$, we clearly have $\as_{K,p}(\psi) = 0$, $1$, or $\eps$, respectively.

In the only remaining case, $p$ is strictly $\eps$-bounded and we have $\lim_{n\to\infty} \mu_{n,p}[ \pi\ext{\psi} \in J ] = 1$ exactly if $(\eps,\delta) \subseteq J$ for some $\delta > \eps$.
We claim that $\eps \ll \gamma$ for every $\gamma > \eps$ (in particular, $(\eps,\delta)$ is non-empty).
This is true by assumption for $\gamma = 1$, so we only consider $\gamma < 1$.
If there would be a smallest $\gamma > \eps$, then $p$ would be weakly $\gamma$ bounded, a contradiction.
If there would be two minimal $\gamma,\gamma' > \eps$, then $\gamma \meet \gamma' = \eps$ and we would be in the case where $\eps \in J$ (see the proof of \cref{infinite-0-1}).
Hence the claim holds.

Assume that the almost sure valuation exists, so $\as_{K,p}(\psi) = j$ for some $j \in K$.
Then there is a sequence of intervals with $\bigcap_{i < \omega} J_i = \{j\}$ and $\lim_{n\to\infty} \mu_{n,p}[ \pi\ext{\psi} \in J_i ] = 1$ for all $i$.
Clearly $j \ge \eps$, as every $J_i$ must contain a non-empty interval $(\eps,\delta)$ for some $\delta > \eps$.
Assume towards a contradiction that $j > \eps$.
By the claim, there is some $j > \gamma' > \eps$.
But then there must be an $i$ such that $\gamma' \notin J_i$, as otherwise $\gamma' \in \bigcap_{i < \omega} J_i$.
Since $J_i$ must contain $j > \gamma'$, this means that $J_i$ cannot intersect $(\eps,\gamma')$.
This leads to a contradiction, since $\lim_{n\to\infty} \mu_{n,p}[ \pi\ext{\psi} \in (\eps,\gamma') ] = 1$ and hence $\lim_{n\to\infty} \mu_{n,p}[ \pi\ext{\psi} \in J_i ] = 0$.
\end{proof}

We remark that the almost sure valuations of sentences can be different in \emph{finite} and \emph{infinite} lattice semirings.
The polynomials are the same, and hence $f_\psi \in \{0,1\}$ implies $\as_{K,p}(\psi) = f_\psi$ in both cases, but the values can differ in case of $f_\psi = e$ if $p$ admits arbitrarily small positive values.

\begin{example}
Consider the semiring $K=([0,1], \max, \min, 0, 1)$ over real numbers.
We define a discrete probability distribution $p \from K^+ \to [0,1]$ by $p(\frac 1 {2^n}) = \frac 1 {2^{n+1}}$ for all $n \in \N$, and $p(x) = 0$ otherwise.
Then $p(1) = \frac 1 2 > 0$ and $p$ is $\epsilon$-bounded for $\eps=0$, as the values $\frac 1 {2^n}$ with positive probability become arbitrarily small.

The sentence $\psi = \A x (Px \lor \neg Px)$ induces $f_\psi = e$ and is clearly (almost surely) true in the Boolean semiring.
However, since $p$ is $0$-bounded, we have $\lim_{n\to\infty}\mu_{n,p}[ \pi\ext{\psi}\in J_i ] = 1$ for the intervals $J_i = [0,\frac 1 {2^i}]$ and hence $\as_{K,p}(\psi) = 0$.
\end{example}

\section{Absorptive Semirings}

We now generalize our results beyond min-max and lattice semirings to more general semirings $(K,+,\cdot,0,1)$ which are \emph{absorptive}, that is, $a+ab = a$ for all elements $a,b \in K$.
Absorption implies idempotence ($a+a=a$ for all $a \in K$) and the semiring is thus partially ordered by the \emph{natural order} $a \leq_K b \Leftrightarrow a+b=b$.
Notice that addition coincides with the supremum $\join$ of the natural order $\le_K$.
In contrast, multiplication can be different from the infimum $\meet$, but is guaranteed to be decreasing due to absorption, i.e., $ab \le a$ and hence $ab \le a \meet b$.
There are many examples of absorptive semirings, including
\begin{itemize}
\item the \emph{Viterbi semiring} $\Vit=([0,1]_\Real,\max,\cdot,0,1)$, used for confidence scores,
\item \emph{tropical semirings}, such as $\Trop= (\Real_{+}^{\infty},\allowbreak \min,\allowbreak +,\allowbreak \infty, 0)$ over the non-negative reals with $\infty$,
\item the \emph{\L{}ukasiewicz semiring} $\Luk = ([0,1]_\Real,\max,\star,0,1)$ with
$x\star y=\max(0,a+b-1)$, used in many-valued logics, and its finite variants, the truncation semirings $T_n = (\{0,\dots,n\}, \max, \star, 0, n)$ with $x \star y = \max(0, x+y-n)$,
\item the \emph{semirings $\Sinf[X]$ of generalised absorptive polynomials} \cite{DannertGraNaaTan21},
\item all min-max-semirings and bounded distributive lattices.
\end{itemize}

With every absorptive semiring $(K,+,\cdot,0,1)$ we can associate the
lattice semiring $\Kmin=(K,+,\meet,0,1)$ over the same domain that replaces multiplication
with the infimum-operation of the natural order (if $K$ is totally ordered, this is simply the minimum).

Let $\pi\from\Lit_A(\tau)\to K$ be a $K$-interpretation into an absorptive semiring $K$.
Since $K$ and $\Kmin$ have the same domain, we can view it also as an interpretation $\pimin$
into $\Kmin$, with $\pimin(\beta)=\pi(\beta)$ for all $\beta\in\Lit_A(\tau)$.
Since also the natural order is the same for $K$ as for $\Kmin$, we can compare
the semiring values $\pi\ext{\psi}$ and $\pimin\ext{\psi}$ for any fully instantiated 
first-order formula and use the observation $ab \le a \meet b$ to lift our results from lattice semirings to absorptive semirings.

\begin{proposition}\label{maxabsorptive_less} For every formula $\psi(\tx)$ and every tuple $\ta$, we have
\begin{itemize}
\item $\fmla \psi \ta  \leq_K  \pimin\ext{\psi(\ta)}$;
\item $\pimin\ext{\psi(\ta)} = 1$ if, and only if, $\fmla \psi \ta = 1$,
\end{itemize}
\end{proposition}

\begin{proof}
The first statement readily follows by induction. For disjunctions and existential quantification the
induction step is trivial, since these are interpreted by the $+$ operation (supremum) in both semirings.
For conjunction, we have
$\pi\ext{\psi\land\phi} = \pi\ext\psi \cdot \pi\ext\phi \leq \pi\ext\psi \meet \pi\ext\phi \leq \pimin\ext\psi \meet \pimin\ext\phi
= \pimin\ext{\psi\land\phi}$
by absorption, analogously for universal quantification.	

For the second statement, it then remains to prove that $\pimin\ext{\psi(\ta)}=1$ implies $\fmla \psi \ta=1$.
The induction step for disjunctions and existential quantification is again trivial.
For conjunctions, observe that $\pimin\ext{\psi \land \phi} = 1$ implies $\pimin\ext{\psi} = \pimin\ext{\phi} = 1$ and hence $\pi\ext{\psi \land \phi} = 1 \cdot 1 = 1$ by induction.
The same argument applies for universal quantification.
\end{proof}

We remark that the equivalence that we have for the value $1$ also holds for $0$ if the semiring has no divisors of $0$, but not in general.
For instance, an interpretation into the \L{}ukasiewicz
semiring $\Luk$ that interprets two literals $\alpha$ and $\beta$ by values in the open interval $(0, \frac 1 2)$, interprets
the conjunction $\alpha\land\beta$ by $0$ whereas the associated interpretation into $\Kmin[\Luk]$ picks
the smaller of the two values.

\Cref{maxabsorptive_less} implies that every first-order sentence $\psi$ whose 
almost sure valuation $\as_{\Kmin,p}(\psi)$  is 0 or 1, has the same almost sure valuation in $K$.
However, if $\as_{\Kmin,p}(\psi)=\epsilon$, with $\epsilon>0$, the situation is more complicated,
since $K$ need not be mulitiplicatively idempotent. Hence, even if $\epsilon$ is the smallest
positive value that may appear for valuations $\pi(\alpha)$ of literals, more complicated
formulae may get smaller valuations.

\begin{Example}[values smaller than $\eps$]
For a simple example, consider a  universal sentence $\psi=\A y( Py \lor \neg Py)$
and a random interpretation in the Viterbi semiring $\Vit=([0,1]_\Real,\max,\cdot,0,1)$, with a
probability distribution $p$ satisfying $p[0<x<1/2]=0$, and $p[x=1/2]>0$. 
Here, $\epsilon=1/2$ and in the associated min-max semiring on $[0,1]$, we clearly
have that $\psi$ asymptotically almost surely evaluates to $\epsilon$.
But for random interpretations in the Viterbi semiring the valuation of $\psi$ asymptotically gets 
arbitrary small; indeed for every $\delta>0$ we have that 
$\lim_{n\to\infty} \mu_{n,p}[ 0<\pi_\Vit\ext\psi <\delta] = 1$ and hence $\as_{\Vit,p}(\psi) = 0$.
\end{Example}

\begin{Example}[absorptive polynomials]
A perhaps more unusual, but also more interesting example is obtained by evaluating the same formula $\psi$ in the semiring
$\Sinf[x]$ of generalised absorptive polynomials with
just one indeterminate (in this case, the natural order is total: $1 > x > x^2 > \dots > x^\infty > 0$)
under the probability distribution that assigns to each pair of complementary literals
$P(j), \neg P(j)$ with equal probability 1/4 pair of values from $\{(1,0),(x,0),(0,1),(0,x)\}$.
Intuitively this means that for each atom, we first decide, independently and with uniform probability whether it is true or false,
and then, again independently and with uniform probability, whether  or not we want to track the effect of this
decision for the valuation of the formulae we consider. The valuation $\pi\ext\phi\in\Sinf[x]$ then is
either 0,1, or a monomial $x^k$, for $k\in\N\cup\{\infty\}$ which tells us, how many tracked literals
are needed for establishing the truth of $\phi$.
For the given probability distribution $p$, we have that $\epsilon=x$, and indeed, for any natural number $n$
and every element $j\in[n]$, the probability measures $\mu_{n,p}$ evaluate the formula $P(j)\lor\neg P(j)$ either to
1 or to $x$, each with probability 1/2. As a consequence, we have for $\psi=\A y( Py \lor \neg Py)$
that $\lim_{n\to\infty} \mu_{n,p}[ 0<\pi\ext\psi < x^k] = 1$, for all $k\in\N$.
Hence the almost sure valuation of $\psi$ is $\as_{\Sinf[x],p}(\psi)=x^\infty$.  
\end{Example}

We can nevertheless show that almost sure valuations of $\eps$ transfer from $\Kmin$ to $K$, under the assumption that $\eps$ is idempotent ($\eps \cdot \eps = \eps$).
This applies, for instance, to the smallest non-zero element $x^\infty$ of $\Sinf[x]$ (and also to the multivariate case, say $x^\infty y^\infty$ in $\Sinf[x,y]$).

\begin{proposition}\label{maxabsorptive_eps_idempotent} Let $\eps \in K$ such that $\eps \cdot \eps = \eps$.
Then for every formula $\psi(\tx)$ and every tuple $\ta$, we have that $ \pimin\ext{\psi(\ta)} = \eps$ 
implies that also $\fmla \psi \ta = \eps$.
\end{proposition}

\begin{proof}
By \cref{maxabsorptive_less}, it suffices to prove by induction on $\psi$ that $\pimin\ext{\psi(\ta)} \ge \eps$ implies $\fmla \psi \ta \ge \eps$.
For literals, disjunctions and existential quantification, this is trivial.
For conjunction, observe that $\pimin\ext{\psi \land \phi} \ge \eps$ implies $\pimin\ext{\psi},\pimin\ext{\phi} \ge \eps$.
Then $\pi\ext{\psi \land \phi} \ge \eps \cdot \eps = \eps$ by induction and monotonicity.
Analogously for universal quantification.
\end{proof}

With this assumption, we can lift the 0-1 laws for finite and infinite lattice semirings to absorptive semirings, leading to the following result about the almost sure valuations.

\begin{corollary}
If $\as_{\Kmin,p}(\psi) \in \{0,1\}$ or $\as_{\Kmin,p}(\psi) \in \{0,1,\eps\}$ with $\eps \cdot \eps = \eps$ (in $K$), then $\as_{K,p}(\psi) = \as_{\Kmin,p}(\psi)$.
\end{corollary}

\section{The natural semiring}

We now discuss the natural semiring $(\N,+,\cdot,0,1)$, which is important for bag semantics in databases.
The most important technical difference to the previously considered semirings is that multiplication in $\N$ is
increasing rather than decreasing with respect to the natural order, which leads to a different asymptotic behaviour of universal quantification.

We first define an extension property adapted to $\N$, which is both stronger and weaker compared to 
the extension properties for lattice semirings:
stronger, since it guarantees not just one, but many realisations of extension types, but
weaker, since it does not guarantee realisations of every type (which would be infinitely many), but only that every underlying Boolean type has  realisations with sufficiently large values.

\begin{Definition}
Given an atomic $m$-type $\rho$ with values in $\N$ we call an extension
$\rho'\in\extend(\rho)$ \emph{large} if out of any pair $\alpha,\neg \alpha$
of complementary literals that contain the variable $x_{m+1}$,
it maps one of them to 0, and to other to some number $\geq 2$.

Recall that $\rho \simbool \rho'$ holds if $\rho$ and $\rho'$ induce the same Boolean type.
Let $\gamma>0$ be some constant.
We say that an $\N$-interpretation $\pi\from\Lit_A(\tau)\to\N$  has the \emph{strong $(k,\gamma)$-extension property}
if for every $m<k$, every tuple $\ta\in A^m$ and every extension $\rho^+ \in \extend(\rho^\pi_\ta)$,
\[    | \{b\in A\setminus\ta :  \rho^\pi_{\ta,b}\simbool \rho^+ \text{ and } \rho^\pi_{\ta,b} \text{ is large}\}| \geq \gamma |A|.\]
\end{Definition}

We consider probability distributions $p\from\N^+\to [0,1]$ with the property that
$p[x\geq 2]>0$ and the associated measures $\mu_{n,p}$ on $\N$-interpretations of $\tau$-structures
with universe $n$. Associated random $\N$-interpretations almost surely have strong extension properties.

\begin{proposition} \label{strong-ext} For any such probability distribution $p$ and every $k\in\N$ there exists some $\gamma>0$ 
such that 
\[ \lim_{n\to\infty} \mu_{n,p}[ \pi \text{ has the strong $(k,\gamma)$-extension property} ] = 1.\]
\end{proposition}

This follows by general results of probability theory that have, for instance 
been also used by  Blass and Gurevich \cite{BlassGur03} to prove strong extension properties of random graphs.  Specifically, we can apply the following fact, see \cite[Lemma 6.2]{BlassGur03}.

\begin{lemma} Let $X$ be the number of successes in $n$ trials, each having at least probability $\delta$ of
success. Then, for each $\alpha\in(0,1)$ there exists some $\beta\in(0,1)$ such that for all natural numbers $n$,
$\text{Prob}[ X \leq \alpha \delta n] \leq \beta^n$.
\end{lemma}

\begin{proof}[Proof of \cref{strong-ext}]
Let $q = p[x \ge 2] > 0$.
For every tuple $\ta\in [n]^m$, and every 
new element $b\in [n]\setminus\ta$, the probability that $\rho^\pi_{\ta,b}$ is large
and  $\rho^\pi_{\ta,b}\simbool \rho^+$  is at least $\delta=(q/2)^\ell$ 
where $\ell$ is the number of relational atoms containing the variable
$x_{m+1}$. Fix any $\alpha\in(0,1)$ and choose $\gamma$ with $0<\gamma < \alpha\delta$.
For large enough $n$, the probabilty that $\rho^\pi_{\ta}$ does not have at at least $\gamma n$
extensions to a type $\rho^\pi_{\ta,b}\geq\rho^+$ is then bounded by $\beta^{n-m}$, for some $\beta<1$.
There are $n^m$ tuples $\ta$ to consider and $2^\ell$ $\simbool$ equivalence classes of
extensions $\rho^+$. Thus the probability that the strong $k$-extension property fails for $\pi$
is bounded by $n^m2^\ell\beta^{n-m}$ which converges to 0 exponentially fast. 
\end{proof}

We again want to represent formulae $\psi(x_1,\dots,x_i)\in\FO^k$ by algebraic expressions
$g(\XX i)$ with indeterminates $X_\alpha$ and $X_{\neg\alpha}$, for each $\tau$-atom
in variables from $x_1,\dots,x_i$.
However, rather than polynomials as used in the case of lattice semirings, we need here a slightly different definition
to include $\infty$ as a coefficient and exponent:

\begin{Definition}
Let $\Ninf:=\N\cup\{\infty\}$.  An \emph{$\infty$-expression} (over $\XX i$) is a formal arithmetic expression $g(\XX i)$ consisting of indeterminates $X_\alpha, X_{\neg\alpha} \in \XX i$, constants $0,1,\infty$, binary operations $+,\cdot$ and the unary operation ${}^\infty$ (infinite power).
Given a mapping $\sigma \from \XX i \to \Ninf$, the $\infty$-expression $g(\XX i)$ evaluates to $g[\sigma] \in \Ninf$ with the usual rules, extended by: $0 \cdot \infty = 0$, $0+\infty=n+\infty=n\cdot\infty=\infty$ for $n \ge 1$ as well as $0^\infty = 0$, $1^\infty = 1$ and $n^\infty = \infty$ for $n \ge 2$.

Two $\infty$-expressions $g(\XX i)$ and $g'(\XX i)$ are \emph{equivalent}, denoted $g \equiv g'$,  if $g[\sigma]=g'[\sigma]$ for every consistent mapping $\sigma \from \XX i \to \Ninf$ (that is, out of any pair $X_\alpha, X_{\neg\alpha}$, one is mapped to 0 and the other one to a non-zero value).
\end{Definition}

We remark that we usually evaluate an $\infty$-expression for a given type $\rho$ mapping literals to $\N$ (not to $\Ninf$).
However, we also consider selector functions into $\Ninf$ and it is thus more convenient to regard types as mappings of the form $\rho \from \XX i \to \Ninf$.
Given such a mapping $\rho$ and a selector function $s \from \YY i \to \Ninf$, we write $\rho s$ for the combined mapping $\rho s \from \XX {i+1} \to \Ninf$ that behaves like $\rho$ on $\XX i$ and like $s$ on $\YY i = \XX {i+1} \setminus \XX i$.

\begin{lemma}
\label{lemInftyexprTrivial}
Let $\sigma,\sigma' \from \XX i \to \Ninf$ and let $g(\XX i)$ be an $\infty$-expression. Then,
\begin{itemize}
\item if $\sigma \simbool \sigma'$, then $g[\sigma] = 0$ if and only if $g[\sigma'] = 0$,
\item if $\sigma \le \sigma'$, then also $g[\sigma] \le g[\sigma']$.
\item if $g$ is not constant on atomic types, then it assumes arbitrarily large values: for every $n\in\N$ there exists a
type $\sigma_n$ with $g[\sigma_n]\geq n$.
\end{itemize}
\end{lemma}

\begin{proof}
The first two claims follow by a straightforward induction on $g$, since
the operations $+,\cdot,^\infty$ are monotone.
For the third claim, assume that there exist atomic types $\sigma,\sigma'$ 
with $g[\sigma]>g[\sigma']$. 
We can then, without loss of generality, choose $\sigma$ and $\sigma'$
so that they differ on precisely one pair $\alpha,\neg \alpha$ of complementary literals,
i.e. $\sigma(\beta)=\sigma'(\beta)$ for all $\beta\not\in\{\alpha,\neg\alpha\}$.
Further we assume that $\sigma(\alpha)>\sigma'(\alpha)$. For each $n$, we then consider
the type $\sigma_n$ such that $\sigma_n(\beta)=\sigma(\beta)$ for all $\beta\neq\alpha$ and
$\sigma_n(\alpha)=\max(n,\sigma(\alpha))$.
We claim that, for every $\infty$-expression $f$ with $f[\sigma]>f[\sigma']$,  we have that $f[\sigma_n]\geq n$.
The only atomic expression $f$ with $f[\sigma]>f[\sigma']$ is $f=X_\alpha$, 
for which $f[\sigma_n]\geq n$.
If $f=f_0+f_1$ then $f_i[\sigma]>f_i[\sigma']$ for $i=0$ or $i=1$, and hence, by induction hypothesis
$f_i[\sigma_n]\geq n$, and hence also $f[\sigma_n]\geq n$. If $f=f_0\cdot f_1$, 
then also  $f_i[\sigma]>f_i[\sigma']$ for $i=0$ or $i=1$,  so $f_i[\sigma_n]\geq n$;
moreover $f_{1-i}[\sigma]>0$ and since $\sigma\leq\sigma_n$, also
$f_{1-i}[\sigma_n]>0$. It follows that $f[\sigma_n]\geq n$.
For $f=\infty\cdot h$ or $f=h^\infty$ we have that $f[\sigma]>f[\sigma']$ implies that
$h[\sigma]>h[\sigma']$ and hence $h[\sigma_n]\geq n$, which implies that $f[\sigma_n]=\infty$.  
\end{proof}

\begin{lemma}
\label{lemInftyexprLarge}
Let $g(\XX i)$ be an $\infty$-expression and consider $\sigma,\sigma' \from \XX i \to \Ninf$ such that $\sigma \simbool \sigma'$ and for all $X_\beta \in \XX i$, either $\sigma(X_\beta) = \sigma'(X_\beta)$ or $\sigma'(X_\beta) \ge 2$.
Then $g[\sigma] \ge 2$ implies $g[\sigma'] \ge 2$.
\end{lemma}

Notice that this lemma applies in particular to $\sigma = \rhopi_{\ta} s$ (the type $\rhopi_\ta$ extended by a selector function $s$) and $\sigma' = \rhopi_{\ta,b}$.
Indeed, if $\rhopi_{\ta,b} \simbool \rhopi_\ta s$ and $\rhopi_{\ta,b}$ is a large extension of $\rhopi_\ta$, the condition in the lemma is satisfied.

\begin{proof}
By induction on $g$.
The claim is trivial for constants.
\begin{itemize}
\item If $g = X_\beta$, then either $g[\sigma] = g[\sigma']$ or $g[\sigma'] \ge 2$, so the claim holds.
\item If $g = g_1 \cdot g_2$ and $g[\sigma] \ge 2$, then w.l.o.g.\ $g_1[\sigma] \ge 2$ and $g_2[\sigma] \neq 0$.
By induction and $\sigma \simbool \sigma'$, the same holds for $\sigma'$ and we have $g_2[\sigma'] \ge 2$.
\item If $g = h^\infty$ and $g[\sigma] \ge 2$, then also $h[\sigma] \ge 2$ and the claim follows by induction.
\item If $g = g_1 + g_2$ and $g[\sigma] \ge 2$, we distinguish two cases.
First assume that $g_1[\sigma] = g_2[\sigma] = 1$.
Since $\sigma \simbool \sigma'$, it follows that $g_1[\sigma'],g_2[\sigma'] \neq 0$ and hence $g[\sigma'] \ge 2$.
Otherwise, w.l.o.g.\ $g_1[\sigma] \ge 2$ and the claim follows by induction.
\qedhere
\end{itemize}
\end{proof}

The definition of the arithmetic expressions $g_\psi(\XX i)$ is to some extent analogous to the 
one for polynomials $f_\psi(\XX i)$ for lattice semirings, but the algebraic operations are no longer idempotent
and the rules for the quantifiers are different and use the constant $\infty$.

\begin{Definition}
Given $\psi(x_1,\dots,x_i)\in\FO(\tau)$ in negation normal form and written with the excluding quantifiers $\Ene$ and $\Ane$, we define the associated $\infty$-expression $g_\psi(\XX i)$ inductively as follows.
\begin{itemize}
\item For (in)equalities, literals, disjunctions and conjunctions, the definition is identical to \cref{def-polynomial}. That is, we set $g_\psi \in \{0,1\}$ for (in)equalities, $g_\psi = X_\alpha$ and $g_\psi = X_{\neg \alpha}$ for (negated) atoms, $g_{\psi\lor\phi}\coloneqq g_\psi + g_\phi$ and $g_{\psi\land\phi}\coloneqq g_\psi \cdot g_\phi$.

\item For $\psi(\tx)=\Ene y\ \phi(\tx,y)$ and $\psi'(\tx)=\Ane y\ \phi'(\tx,y)$, let $\YY i=\XX{i+1}\setminus \XX i$.
Let further $S$ be the set of all consistent selector functions $s \from \YY i \to \{0,\infty\}$ (that is, one of $X_\alpha$, $X_{\neg \alpha}$ is mapped to $0$, the other one to $\infty$, for every atom $\alpha$).
Now set
\[
    g_\psi(\XX i) \coloneqq \infty \cdot \left( \sum_{s \in S} g_\phi(\XX i,s(\YY i)) \right), \quad
    g_{\psi'}(\XX i) \coloneqq \left( \prod_{s \in S} g_{\phi'}(\XX i,s(\YY i)) \right)^\infty.
\]
\end{itemize}
\end{Definition}

Notice that a single positive value in the sum or a single value $\ge 2$ in the product will result in the value $\infty$.
This is justified by the $(k,\gamma)$-extension property, which guarantees that every selector function has not just one, but many large realisations which, as $n$ grows, lead to arbitrarily large values of $\psi$ or $\psi'$.

\begin{Example}
Recall $\psi = \Ene x(\neg Exx\land\Ane y(Exy\lor (\neg Exy\land \Ene z(Exz\land Ezy))))$ of \cref{exPolynomials}.
Using the same notation, we obtain the following $\infty$-expressions (we always simplify expressions without indeterminates, e.g.\ $\infty \cdot (\infty + 0) = \infty$ in the first step).

\begin{center}
\begin{tabular}{|c|c|}
\hline
$\phi$&$g_\phi$\\ \hline\hline
$Exz \land Ezy$&$ZU$\\
$\Ene z ( Exz \land Ezy)$& $\infty$\\ 
$\neg Exy\land \Ene z(Exz\land Ezy)$& $\nn{Y} \cdot \infty$\\
$Exy\lor (\neg Exy\land \Ene z(Exz\land Ezy))$& $Y + \nn{Y} \cdot \infty$\\
$\Ane y(Exy\lor (\neg Exy\land \Ene z(Exz\land Ezy))))$&$\infty$\\
$\neg Exx\land\Ane y(Exy\lor (\neg Exy\land \Ene z(Exz\land Ezy))))$&$\nnX \cdot \infty$\\
$\psi$&$\infty$\\
\hline
\end{tabular}
\end{center}

It should come as no surprise that the resulting value is positive, since $\psi$ is asymptotically almost surely true in Boolean semantics.
But notice that \cref{exPolynomials} resulted in $e$, representing the smallest positive value, whereas we obtain the largest value $\infty$ in the natural semiring.
\end{Example}

\begin{Example}
For an example with more complicated $\infty$-expressions, consider $\psi = \Ene x (Px \land \Ane y (x \neq y \lor \neg P x))$.
Using the indeterminate $X$ for $Px$, we obtain:

\begin{center}
\begin{tabular}{|c|c|}
\hline
$\phi$&$g_\phi$\\ \hline\hline
$x \neq y \lor \neg P x$ & $1 + \nnX$\\
$\Ane y (x \neq y \lor \neg P x)$ & $(1 + \nnX)^\infty$\\ 
$Px \land \Ane y (x \neq y \lor \neg P x)$& $X \cdot (1 + \nnX)^\infty$\\
$\psi$& $\infty$\\
\hline
\end{tabular}
\end{center}

For an example resulting in $0$, replace the subformula $x \neq y$ by $x = y$ to obtain $X \cdot (0+\nnX)^\infty$ in the third and thus $0$ in the last row.
\end{Example}

The main technical result of this section is the following theorem, similar to \cref{ext-polynomials,delta-extension}.
We again prove that the expressions $g_\psi$ provide an adequate description of formulae $\psi(\tx)$, now including the special case where the value of $\psi$ becomes arbitrarily large.

\begin{theorem}\label{large-extension} Let 
$\gamma>0$ and let $\pi\from\Lit_{[n]}(\tau)\to \N$ be an $\N$-interpretation
on universe $[n]$ with the strong $(k,\gamma)$-extension property.
Further assume that $\gamma n\geq 2$.
Then, for every formula $\psi(x_1,\dots,x_i)\in\FO^k(\tau)$ with the
associated $\infty$-expression $g_\psi$ and every tuple $\ta$ of distinct elements from $[n]$ either\\
{\rm (1)}\quad  $\fmla \psi \ta = g_\psi[\rho^\pi_{\ta}]\in\N$, or\\
{\rm (2)}\quad  $g_\psi[\rho^\pi_{\ta}]=\infty$ and $\fmla \psi \ta \geq \gamma n$.
\end{theorem}

\begin{proof}
We proceed by induction on $\psi$. If $\psi$ is a literal, it is immediate from
the definition of $g_\psi$ that case (1) holds.

\medskip\noindent For $\psi=\phi\lor\theta$ we have $g_\psi\coloneqq g_\phi+g_\theta$. 
If case (1) holds for both $\phi$ and $\theta$ then also for $\psi$.
Otherwise case (2) applies for $\phi$ or $\theta$, and then obviously also for $\psi$.
The argument for $\psi=\phi\land\theta$ is analogous (taking into account the case that one
of formulae evaluates to 0).

\medskip\noindent Let now $\psi(\tx)=\Ene y\ \phi(\tx,y)$.
We will show that case (1) applies if $\fmla \psi \ta = \exparho \psi \ta = 0$, otherwise case (2) applies.
We recall:
\[
  \fmla \psi \ta = \sum_{b \in A \setminus \ta} \fmla \phi {\ta,b}, \qquad
  \exparho \psi \ta = \infty \cdot \left( \sum_{s \in S} \expa \phi {\rhopi_\ta, s(\YY i)} \right).
\]

We begin with a general observation that is used throughout the proof:
For every element $b \in A \setminus \ta$, we can define a selector function $s_b \in S$ with $\rhopi_{\ta,b} \simbool \rhopi_\ta s_b$ by simply setting $s_b(X_\beta) = \infty$ precisely if $\rhopi_{\ta,b}(\beta) > 0$.
Conversely, every selector function $s \in S$ is consistent, so there is a type $\rho^+ \in \extend(\rhopi_\ta)$ with $\rho^+ \simbool \rhopi_\ta s$ (we may set $\rho^+(\beta) = 1$ whenever $s(X_\beta) = \infty$).
The $(k,\gamma)$-extension property then guarantees at least $\gamma n$ many elements $b \in A \setminus \ta$ such that $\rhopi_{\ta,b} \simbool \rhopi_\ta s$ and $\rhopi_{\ta,b}$ is large.

First assume that there is a selector function $s \in S$ with $\expa \phi {\rhopi_\ta,s(\YY i)} > 0$.
Then obviously $\exparho \psi \ta = \infty$.
There are at least $\gamma n$ many elements $b$ with $\rhopi_{\ta,b} \simbool \rhopi_\ta s$ as descried above.
By \cref{lemInftyexprTrivial}, this implies $\exparho \phi {\ta,b} > 0$.
Then also $\fmla \phi {\ta,b} > 0$ by induction (using either case (1) or (2)) and hence $\fmla \psi \ta \ge \gamma n$ for the sum, so case (2) holds.

Now assume that no such selector function exists, hence $\exparho \psi \ta = 0$.
If there was a $b$ with $\fmla \phi {\ta,b} > 0$, then $\exparho \phi {\ta,b} > 0$ by induction and the selector function $s_b$ contradicts our assumption.
Hence no such $b$ exists and we have $\fmla \psi \ta = 0$, so case (1) applies.

\medskip\noindent 
Finally, let $\psi(\tx)=\Ane y\,\phi(\tx,y)$.
Similarly to the previous case, we will show that case (1) only applies for the values $0,1$, otherwise case (2) applies.
We recall:
\[
  \fmla \psi \ta = \prod_{b \in A \setminus \ta} \fmla \phi {\ta,b}, \qquad
  \exparho \psi \ta = \left( \prod_{s \in S} \expa \phi {\rhopi_\ta, s(\YY i)} \right)^\infty.
\]

We proceed with a similar case distinction, but additionally account for one of the factors being $0$.
To this end, assume there is $s \in S$ with $\expa \phi {\rhopi_\ta, s(\YY i)} = 0$, hence also $\exparho \psi \ta = 0$.
We obtain an element $b$ with $\rhopi_{\ta,b} \simbool \rhopi_\ta s$ by the extension property.
By \cref{lemInftyexprTrivial} and case (1), this yields $0 = \exparho \phi {\ta,b} = \fmla \phi {\ta,b} = \fmla \phi \ta$, so case (1) applies.
Similarly, if $\fmla \phi {\ta,b} = 0$ for some $b$, then $\exparho \phi {\ta,b} = 0$ by case (1) and $\expa \phi {\rhopi_\ta, s_b(\YY i)} = 0$ for the induced selector function.

From now on, we thus have $\fmla \phi {\ta,b},\expa \phi {\rhopi_\ta, s(\YY i)} > 0$ for all $s \in S$.
First assume that there is a selector function $s$ with $\expa \phi {\rhopi_\ta, s(\YY i)} \ge 2$.
Then $\expa \psi \ta \ge 2^\infty = \infty$.
The extension property guarantees $\gamma n$ many elements $b$ such that $\rhopi_{\ta,b}$ is large and $\rhopi_{\ta,b} \simbool \rhopi_\ta s$.
All of these elements satisfy $\exparho \phi {\ta,b} \ge 2$ by \cref{lemInftyexprLarge} and thus $\fmla \phi {\ta,b} \ge 2$ by induction (using either case (1) or case (2)).
We thus have $\fmla \psi \ta \ge 2^{\gamma n} \ge \gamma n$ and case (2) applies.

Lastly, the only remaining case is that $\expa \phi {\rhopi_\ta, s(\YY i)} = 1$ for all selector functions.
Then also $\expa \psi \ta = 1^\infty = 1$.
If there was a $b$ with $\fmla \phi {\ta,b} \ge 2$, then also $\exparho \phi {\ta,b} \ge 2$ (using either case (1) or (2)).
By \cref{lemInftyexprLarge}, the induced selector function $s_b$ contradicts our assumption: $\expa \phi {\rhopi_\ta, s(\YY i)} \ge 2$.
Hence $\fmla \phi {\ta,b} = 1$ for all $b$ and thus $\fmla \psi \ta = 1$ as well.
\end{proof}

\begin{corollary}[0-1 law for $\FO$ on the natural semiring] \label{0-1-law-naturla}
Let $p\from\N\setminus\{0\}\to [0,1]$ be a probability distribution with $p[x\geq 2]>0$
and let $\tau$ be a relational vocabulary. Then for every sentence $\psi\in\FO(\tau)$
there either is a value $j\in\N$ such that $\lim_{n\to\infty}\mu_{n,p}[ \pi\ext{\psi}=j]=1$,
or  $\lim_{n\to\infty}\mu_{n,p}[ \pi\ext{\psi}>j]=1$ for all $j\in\N$.
\end{corollary}

\begin{proof}
For every sentence $\psi$ the associated $\infty$-expression $g_\psi$ contains no indeterminates
and thus evaluates to a value $g_\psi[\emptyset] \in \Ninf$.
Since random $\N$-interpretations almost surely have the strong $(k,\gamma)$-extension
property, for all $k,\gamma$, it follows that
$\lim_{n\to\infty}\mu_{n,p}[ \pi\ext{\psi} = g_\psi[\emptyset]]=1$ in case $g_\psi[\emptyset]\in\N$, or
$\lim_{n\to\infty}\mu_{n,p}[ \pi\ext{\psi}>j]=1$ for all $j\in\N$, in case $g_\psi[\emptyset]=\infty$.
\end{proof}

The 0-1 law induces a partition of the sentences $\FO(\tau)$ into classes $(\Phi_j)_{j\in\Ninf}$.
For $j \in \N$, the class $\Phi_j$ contains those sentences $\psi$ for which
$\lim_{n\to\infty}\mu_{n,p}[ \pi\ext{\psi}=j]=1$ (and hence $\lim_{n\to\infty}\mu_{n,p}[ \pi\ext{\psi} > j]=0$).
The additional class $\Phi_\infty$ contains the sentences $\psi$ which almost surely evaluate to unboundedly large values, so $\lim_{n\to\infty}\mu_{n,p}[ \pi\ext{\psi}>j]=1$ for all $j\in\N$.
We first observe that these classes align with the Boolean case.

\begin{lemma}
For every relational first-order sentence $\psi$, we have $\psi \in \Phi_0$ if, and only if, $\psi$ is asymptotically almost surely false under Boolean semantics.
\end{lemma}
\begin{proof}
It is a general observation (see, e.g.\ \cite{GraedelTan17}) that semiring semantics in $\N$ (or any positive semiring) and Boolean semantics are compatible in the sense that for any $\N$-interpretation $\pi$ over universe $[n]$ and any sentence $\psi$, we have $\pi \ext \psi > 0$ if, and only if, $\AA_\pi \models \psi$.
Here, $\AA_\pi$ is the Boolean structure on $[n]$ induced by $\pi$ (i.e.\ $\AA \models \alpha$ iff $\pi(\alpha) > 0$, for all literals $\alpha$).

Recall that for the measure $\mu_{n,p}$ over the semiring $\N$, we first randomly decide for each relational atom whether $\alpha$ or $\neg \alpha$ shall be true (say with probability $\frac 1 2$), and then assign a random positive value from $\N \without 0$ to the true literal.
It thus makes no difference whether we consider random Boolean structures or the Boolean structures induced by random $\N$-interpretations.
That is, $\mu_{n,\frac 1 2}(\psi) = \mu_{n,p}[\pi \ext \psi > 0]$.
Hence also $\lim_{n\to\infty} \mu_{n,\frac 1 2}(\psi) = \lim_{n\to\infty} \mu_{n,p}[\pi \ext \psi > 0]$ and the claim follows.
\end{proof}

Moreover, there are trivial examples showing that all classes $\Phi_j$ are non-empty.
Trivially false sentences such as $\E x(x\neq x)$ are in $\Phi_0$, whereas $\E x(x=x)\in\Phi_\infty$,
and for any $j\in\N\setminus\{0\}$, we have that $\bigvee_{1\leq i\leq j} \A x_i(x_i=x_i) \in\Phi_j$.
However, for all $j\not\in\{0,\infty\}$ there are, in a sense, \emph{only} trivial examples of sentences in $\Phi_j$,
whereas all ``interesting'' sentences are either in $\Phi_0$ or in $\Phi_\infty$, i.e.
are almost surely false, or almost surely have unboundedly large truth values.
To make this precise, we introduce the following notion of trivial formulae.

\begin{Definition}
All formulae $\phi(\tx)$ with $g_\phi\equiv 0$ are called almost surely false. 
The class of \emph{trivial} formulae in $\FO(\tau)$
is defined by induction:
\begin{itemize}
\item Every formula of form $x=x$ or $x\neq y$ (for distinct variables $x,y$) is trivial.
\item Conjunctions of trivial formulae are trivial.
\item Disjunctions of trivial formulae with almost surely false formulae are trivial.
\item Formulae of form $\Ane x \; \psi$ are trivial, if $\psi$ is trivial.
\end{itemize}
\end{Definition} 

That we call such formulae ``trivial'' should be taken with a grain of salt.
They are built from trivial equalities and inequalities, but with the additional building block
of almost surely false formulae, and it is not really trivial (but a \pspace-complete problem) to decide
whether a given formula is almost surely false. Obviously, $g_\psi\equiv 1$ for
all trivial formulae $\psi$. In particular, 
all trivial sentences are in $\Phi_1$.

We next observe that the only way to build a sentence that is neither almost surely
false, nor evaluates to unboundedly large values, is to combine trivial sentences
by disjunctions and conjunctions.

\begin{proposition}
If $\psi\in\Phi_j$ for $j\not\in\{0,\infty\}$, then $\psi$ is a positive Boolean combination
of trivial sentences.
\end{proposition}

\begin{proof}
Recall that $\psi\in\Phi_j$ if $g_\psi[\emptyset] = j$.
Clearly, via quantification we can only produce sentences
whose associated  $\infty$-expressions are equivalent  to
0,1, or $\infty$ (due to multiplication or exponentiation by $\infty$).
The only way to obtain $g_\psi[\emptyset] = j$, with $1<j<\infty$,
is therefore by addition and multiplication of $\infty$-expressions with $g_\phi[\emptyset] = 1$.
We thus have to show that every sentence $\psi$ with $g_\psi[\emptyset] = 1$ or, equivalently, $g_\psi \equiv 1$, must be trivial.
We proceed by induction, to prove this not just for sentences, but for every formula $\psi(\tx)$.

If $\psi=\phi\lor\theta$, then $g_\psi=g_\phi+g_\theta \equiv 1$ implies that
$g_\phi\equiv 1$ and $g_\theta \equiv 0$ (or vice versa), because
otherwise,  either $g_\phi$ and $g_\theta$ would assume
arbitrarily large values by \cref{lemInftyexprTrivial}.
Hence one of the two subsentences $\phi$ or $\theta$ 
must be almost surely false and the other one trivial, hence $\psi$ is trivial as well.
Similarly, if $\psi=\phi\land\theta$, then $g_\psi=g_\phi\cdot g_\theta \equiv 1$ implies that
$g_\phi\equiv g_\theta\equiv 1$ so both $\phi$ and $\theta$ must be trivial, and hence also $\psi$.
It is impossible that $\psi=\Ene y\ \phi$, since in that case $g_\psi$ evaluates to 0 or $\infty$.

Finally, if $\psi(\tx)=\Ane y\,\phi(\tx,y)$, then $g_\psi=(\prod_{s\in S} g_\phi(\XX i,s(\YY i))^\infty$.
By assumption $g_\psi\equiv 1$, so $g_\phi(\XX i,s(\YY i)\equiv 1$ for all $s$.
We claim that then $g_\phi\equiv 1$. If not, then then there is some consistent mapping $\rho \from \XX {i+1} \to \Ninf$ with $g_\phi[\rho]\neq 1$.
Let $\rho_0 = \rho \upharpoonright \XX i$ be the restriction of $\rho$ to $\XX i$ and let $s_\rho \in S$ be the selector function induced by $\rho$.
That is, $s_\rho(X_\beta) = 0$ if $\rho(\beta)=0$ and $s_\rho(X_\beta)=\infty$ otherwise.
For the combined mapping $\rho_0 s_\rho \from \XX {i+1} \to \Ninf$, we then have $\rho \le \rho_0 s_\rho$.
Hence, \cref{lemInftyexprTrivial} implies that if $g_\phi[\rho]=0$ then
also $g_\phi[\rho_0 s_\rho] = 0$, and if $g_\phi[\rho]\geq 2$, also $g_\phi[\rho_0 s_\rho]\geq 2$.
It follows that $g_\phi[\rho_0 s_\rho] \neq 1$, so $g_\phi(\XX i, s_\rho(\YY i)) \not\equiv 1$ (witnessed by $\rho_0$), contradiction.
Hence, we have established that $g_\phi \equiv 1$, and by induction hypothesis, 
it follows that $\phi$ must be trivial.  
Thus, also $\psi=\Ane y\ \phi$ is trivial.
\end{proof}

\begin{corollary} Let $\psi$ be a relational first-order sentences that is not trivial. 
Then either $\psi\in \Phi_0$ or $\psi \in \Phi_\infty$.
\end{corollary}

\medskip\noindent{\bf Remark. }
Instead of the almost sure valuation of a first-order sentence, one might also consider the
\emph{asymptotic expected valuation} $E_p[\pi\ext{\psi}]\coloneqq \lim_{n\to\infty} \sum_{j\in\N} j\cdot \mu_{n,p}[\pi\ext{\psi}=j]$.
However, due to the possibility of extremely large values of particular events with very low probability,
we lose the correspondence to Boolean semantics.
As an example consider the sentence
$\Delta\coloneqq \Ane x \Ane y \Ane z (E xy \land Eyz \land Ezx)$
on random graphs, saying that any three distinct nodes form a triangle.
Clearly this sentence is almost surely false on finite graphs since it only evaluates to a positive value
on cliques. However, for any probability distribution $p$ with $p[x \ge 2]=q>0$, we have that
with probability $q^{n(n-1)/2}$ a random $\N$-valued graph with vertex set $[n]$ is a clique where each edge
has a value $\geq 2$. On such a clique, the value of $\Delta$ is at least $8^{n(n-1)(n-2)}$.
Hence, although $\Delta$ evaluates to 0 on all non-cliques, we have that
$E_p[\pi\ext{\Delta}] \geq \lim_{n\to\infty} q^{n(n-1)/2}\cdot  8^{n(n-1)(n-2)} =\infty$.

\section{The random countable $K$-interpretation}

A classical fact about 0-1 laws in Boolean semantics is the $\omega$-categoricity of the theory $T$ of extension axioms,
and thus the existence of a unique countable $\tau$-structure that satisfies all of them.
For finite semirings, this fact extends in a straightforward way to our setting.

\begin{theorem} For every finite semiring $K$ and every  finite relational vocabulary $\tau$ there exists
a countable $K$-interpretation $\pi_{\mathcal R}\from\Lit_A(\tau)\to K$ that has the $k$-extension property for
all natural numbers $k$. Moreover $\pi_{\mathcal R}$ is unique up to isomorphism.
\end{theorem}

\begin{proof}
\newcommand{\fmlKInt}{\mathsf{Int}_K}
\newcommand{\fmltype}{\mathsf{type}}
For any fixed finite semiring $K$, we can represent $K$-interpretations $\pi\from\Lit_A(\tau)\to K$
as classical relational structures $\AA^\pi$ with universe $A$ over a vocabulary $\tau^K$ consisting
of relations
$R^+_j$ and $R^-_j$, for $R\in\tau$ and $j\in K$, where $R^+_j=\{\ta: \pi(R\ta)=j\}$ and  
$R^-_j=\{\ta: \pi(\neg R\ta)=j\}$. It is then not difficult to axiomatise the extension properties of
$K$-interpretations in $\FO(\tau^K)$:
\begin{itemize}
\item There is a sentence $\fmlKInt\in\FO(\tau^K)$ such that
$\AA\models \fmlKInt$ if, and only if, $\AA\cong\AA^\pi$ for a $K$-interpretation $\pi$.
\item For any atomic $k$-type $\rho$ there is a formula $\fmltype_\rho(\tx)$
such that $\AA^\pi\models \fmltype_\rho(\ta)$ if, and only if, $\rho^\pi_{\ta}=\rho$.
\item Hence $\pi$ has the $k$-extension property if, and only if, 
\[  \AA^\pi\models \A\tx( \fmltype_\rho(\tx)\ra \E y \ \fmltype_{\rho^+}(\tx,y))\]
for every $i\leq k$, every atomic $i$-type $\rho$ and every extension $\rho^+\in\extend(\rho)$. 
\end{itemize}
Let now $T(K)$ be the collection of all these extension axioms together with the sentence $\fmlKInt$.
Obviously, every finite subset of $T(K)$ is satisfiable, so by compactness and the Löwenheim-Skolem Theorem,
there exists a countable model ${\mathcal R}\models T(K)$. It follows that
there exists a countable $K$-interpretation $\pi_{\mathcal R}\from\Lit_A(\tau)\to K$, the one represented by $\mathcal R$,
which has the $k$-extension property for all natural numbers $k$.

Further, the standard back-and-forth argument shows that any two such interpretations must be isomorphic.
Specifically suppose that $\pi_A\from\Lit_A(\tau)\to K$ and  $\pi_B\from\Lit_B(\tau)\to K$, with countable universes $A$ and $B$,
both have the $k$-extension property for all $k\in\omega$. Fix enumerations of the universes $A$ and $B$, 
and construct partial isomorphisms $p_n=\{(a_1,b_1),\dots,(a_n,b_n)\}\subseteq A\times B$ by
induction as follows. Let $p_0=\emptyset$. If $p_n$ is already defined, let $\rho_n$ be the 
$n$-type realised by $\ta=(a_1,\dots,a_n)$ in $\pi_A$, and also by  $\tb=(b_1,\dots,b_n)$ in $\pi_B$
(given that $p_n$ is a partial isomorphism). 
For even $n$, let  $c$ be the first element in the enumeration of $A$ that does not appear in $p_n$,
and let $\rho^+\in\extend(\rho_n)$ be the type realised by $(\ta,c)$ in $\pi_A$. Since $\pi_B$
has the $n$-extension property it follows that there exist some $d\in B$ such that also
$(\tb,d)$ realises $\rho^+$ in $\pi_B$. Select the smallest such $d$ in the enumeration of $B$
and set $p_{n+1}=p_{n}\cup\{(c,d)\}$. For odd $n$, we proceed analogously, starting with
the first $d$ in the enumeration of $B$ that does not occur in $p_n$.
In this way we get an increasing sequence $(p_n)_{n\in\omega}$ of partial isomorphisms
whose union $p\coloneqq \bigcup_{n\in\omega} p_n$ covers all elements of $A$ and $B$ and thus 
is an isomorphism between $\pi_A$ and $\pi_B$.  
\end{proof}

For finite semirings $K$ in which infinite sums and products are well-defined,
the countable random $K$-interpretation provides evaluations $\pi_{\mathcal R}\ext{\psi}$
for arbitrary first-order sentences. In particular, this is the case for lattice semirings.
Notice that Theorem~\ref{ext-polynomials} does not depend on the universe being finite,
which implies that for any finite lattice semiring and every first-order sentence
$\psi\in\FO(\tau)$, we have that $\pi_{\mathcal R}\ext{\psi}=f_\psi$.
But this coincides with the almost sure valuation of $\psi$ on random finite $K$-interpretations.

\begin{corollary}
Let $K$ be a finite lattice semiring with a 
probability distribution $p\from K^+\to(0,1]$, and let $\tau$ be a relational vocabulary.
Then, for every sentence $\psi\in\FO(\tau)$, the valuation 
of $\psi$ by the random countable $K$-interpretation $\pi_{\mathcal R}$ coincides with the almost sure valuation
of $\psi$ by finite $k$-interpretations: $\pi_{\mathcal R}\ext{\psi}=\as_{K,p}(\psi)$. 
\end{corollary}

\section{Conclusion}

We have seen that the most fundamental result on logic on random structures, the 0-1 law for first-order logic,
can be extended to semiring semantics. 
The specific results, and also the proofs, depend on the underlying semiring, but generally follow 
the same pattern. The cornerstone of classical 0-1 laws, the \emph{extension axioms}, generalise to \emph{extension properties} of random semiring interpretations. A new ingredient
is the algebraic representation of first-order formulae by polynomials (or in the case of the natural semiring, by
$\infty$-expressions). The extension properties permit
us to do this with a constant supply of variables, and to obtain for each formula a fixed expression
that is independent of the size of the universe. Besides the generalisation of classical 0-1 laws to results saying
that the asymptotic probabilities of statements $\pi\ext\psi=j$ converge to 0 or 1,
we additionally get here results telling us which values of the semiring can actually appear as
almost sure valuations of first-order sentences. In finite or infinite lattice semirings these are
just three values, 0, 1 and the infimum of all values $j>0$, whereas in the natural semiring
there are rather trivial constructions showing that every number $j\in\N$ can possibly occur
as an almost sure valuation.
We have also studied the complexity of
computing almost sure valuations over finite lattice semirings and proved that this is a \pspace-complete problem.

The results presented here are a first, but fundamental, step towards understanding the power of semiring semantics 
for random interpretations. Indeed we have considered here only the case of random interpretations that are
induced by a \emph{fixed probability distribution} on the semiring, which is independent of the size of the universe.
This corresponds to the $G_{n,p}$-model of random graph theory where $p$ is a constant,
and to the classical 0-1 law of Glebskii et al. and Fagin. Of course, the study of logic on random graphs and random structures has gone beyond that and has, in particular,  investigated models where the probabilities are given by a 
function of the universe, often involving sparse structures, and has for instance studied issues of phase transitions. 
The calculation of probabilities becomes more involved in such cases
and uses much more sophisticated mathematical machinery. The study of semiring semantics for such
more general random models poses an interesting challenge. This will also be relevant for 
applications, because random models arising in practice are in general not given by constant probability distributions.
While classical 0-1 laws give a simple high-level argument for the inexpressibility of Boolean properties
for which the 0-1 law fails, our results may pave the way towards inexpressibility results for numerical parameters
in semiring semantics by showing that their probabilistic behaviour is different from those of logical sentences,
for instance in the natural semiring.
A further interesting aspect is the study of \emph{certain answers}  for queries over incompletely specified databases.
Libkin  \cite{Libkin18} proposes a probabilistic approach that measures how close an answer is to certainty,
based on the observation that for the standard model of missing data, the classical 0-1 law holds.
Semiring semantics, for instance via its connection to bag semantics, confidence scores and cost analysis,
provides an interesting possibility to extend such approaches to a more general setting.

\bibliography{zeroone}
\bibliographystyle{plainurl}

\end{document}